\documentclass[aps,prl,superscriptaddress,twocolumn,10pt]{revtex4-2}

\usepackage[centering,hmargin=1.7cm,vmargin=0.78in]{geometry}

\usepackage{amsmath,mathtools,amsthm,amssymb}
\usepackage{tabularx}
\usepackage{tabularray}

\usepackage{graphicx}

\usepackage{color}
\usepackage{bbold}
\usepackage{enumitem}
\usepackage{centernot}
\usepackage{physics}
\usepackage{comment}
\usepackage{array}
\usepackage{multibib}
\newcites{A}{ReferencesB}
\usepackage{graphics}
\usepackage{wrapfig}
\usepackage{fontsize}
\graphicspath{{figures/}}
\usepackage{bbm}
\usepackage{dsfont}
\usepackage{mathrsfs}
\usepackage{mathdots}
\usepackage{enumitem}
\usepackage{pifont}
\usepackage[table]{xcolor}
\usepackage{booktabs}
\usepackage{pifont}
\usepackage{amssymb}

\DeclareMathOperator{\de}{d\!}

\DeclareMathOperator{\cl}{Cl}

 \newcommand{\even}{\mathrm{Even}(\mathbb{F}_{2}^{k\times m})}
\newcommand{\symf}{\mathrm{Sym}(\mathbb{F}_2^{m\times m})}

\newcommand{\cmark}{\ding{51}}%
\newcommand{\xmark}{\ding{55}}%
\definecolor{lightred}{RGB}{243,229,231}
\definecolor{lightgreen}{RGB}{241,255,239}
\definecolor{lightblue}{RGB}{232,240,244}
\definecolor{RoyalBlue}{RGB}{65,105,225}
\definecolor{ForestGreen}{RGB}{34,139,34}   
\definecolor{Maroon}{RGB}{135,0,0}
\definecolor{myrefcolor}{rgb}{0.067,0.5,0.5}
\definecolor{myurlcolor}{rgb}{0.1,0,0.9}

\usepackage[
    breaklinks,
    pdftex,
    colorlinks=true,
    linkcolor=myrefcolor,
    citecolor=myrefcolor,
    urlcolor=myrefcolor
]{hyperref}
\usepackage[capitalize]{cleveref}
\DeclareMathOperator{\poly}{poly}

\DeclareDocumentCommand\mel{ s s m m m }
{ 
    \IfBooleanTF{#1}
    {
        \IfBooleanTF{#2}
        {\left\langle{#3}\middle\vert{#4}\middle\vert{#5}\right\rangle} 
        {\vphantom{#3#4#5}\left\langle\smash{#3}\middle\vert\smash{#4}\middle\vert\smash{#5}\right\rangle} 
    }
    {\vphantom{#4}\left\langle{#3}\middle\vert\smash{#4}\middle\vert{#5}\right\rangle} 
}

\newtheorem*{theorem*}{Theorem}
\newtheorem{theorem}{Theorem}
\newtheorem{lemma}{Lemma}

\newtheorem{definition}{Definition}

\newtheorem{corollary}{Corollary}

\theoremstyle{remark}

\newtheorem{fact}{Fact}

\usepackage{algorithm}
\usepackage{algpseudocode}

\newcommand{\stab}{\operatorname{STAB}}

\newcommand{\be}{\begin{equation}\begin{aligned}\hspace{0pt}}
\newcommand{\ee}{\end{aligned}\end{equation}}
\newcommand{\ba}{\begin{eqnarray}}
\newcommand{\ea}{\end{eqnarray}}

\newcommand{\haar}[0]{\operatorname{Haar}}

\definecolor{airforceblue}{rgb}{0.36, 0.54, 0.66}

\newcommand{\bb}{\begin{equation}\begin{aligned}\hspace{0pt}}
\newcommand{\bbb}{\begin{equation*}\begin{aligned}}
\newcommand{\eb}{\end{aligned}\end{equation}}
\newcommand{\eeb}{\end{aligned}\end{equation*}}
\allowdisplaybreaks
\begin{document}

\title{Adaptively secure unitary designs with constant non-Clifford cost}



\author{Lennart Bittel}
\thanks{These authors contributed equally. \href{mailto:l.bittel@fu-berlin.de}{l.bittel@fu-berlin.de}, \href{mailto:loleone@unisa.it}{loleone@unisa.it}}

\affiliation{Dahlem Center for Complex Quantum Systems, Freie Universit\"at Berlin, 14195 Berlin, Germany}
\author{Lorenzo Leone}
\thanks{These authors contributed equally. \href{mailto:l.bittel@fu-berlin.de}{l.bittel@fu-berlin.de}, \href{mailto:loleone@unisa.it}{loleone@unisa.it}}


\affiliation{Dipartimento di Ingegneria Industriale, Università degli Studi di Salerno, Via Giovanni Paolo II, 132, 84084 Fisciano (SA), Italy}
\affiliation{Dahlem Center for Complex Quantum Systems, Freie Universit\"at Berlin, 14195 Berlin, Germany}


\begin{abstract}
Randomness is a fundamental resource in quantum information, with crucial applications in cryptography, algorithms, and error correction. A central challenge is to construct unitary $k$-designs that closely approximate Haar-random unitaries while minimizing the costly use of non-Clifford operations. In this work, we present a protocol able to generate unitary $k$-designs on $n$ qubits, secure against any adversarial quantum measurement, with a system-size-independent number of non-Clifford gates. Our construction applies a $k$-design only to a subsystem of size $\Theta(k)$, independent of $n$. This ``seed'' design is then ``diluted'' across the entire $n$-qubit system by sandwiching it between two random Clifford operators. The resulting ensemble forms an $\varepsilon$-approximate unitary $k$-design on $n$ qubits. We prove that this construction achieves full quantum security against adaptive adversaries using only $\tilde{O}(k^2 \log\varepsilon^{-1})$ non-Clifford gates. If one requires security only against polynomial-time adaptive adversaries, the non-Clifford cost decreases to $\tilde{O}(k + \log^{1+c} \varepsilon^{-1})$. This is optimal, since we show that at least $\Omega(k)$ non-Clifford gates are required in this setting. Compared to existing approaches, our method significantly reduces non-Clifford overhead while strengthening security guarantees to adaptive security as well as removing artificial assumptions between $n$ and $k$. 
These results make high-order unitary designs practically attainable in near-term fault-tolerant quantum architectures.
\end{abstract}
\maketitle

Random unitary operations play a central role in modern quantum information science, with applications ranging from cryptography~\cite{ambainis_smaapp2004,kretschmer_quapru2021} and algorithms~\cite{sen_ranmea2006, brandao_expspe2013} to quantum sensing~\cite{kueng_distinguishing_2016}, verification methods~\cite{Eisert_2020}, and communication protocols~\cite{devetak_relqua2004,groisman_quacla2005}. They are equally important as conceptual tools: ensembles of random unitaries often provide effective coarse-grained models for the dynamics of highly complex quantum systems. This perspective has revealed deep insights into spectral statistics of many-body Hamiltonians, mechanisms of thermalization~\cite{popescu_entanglement_2006}, holographic dualities~\cite{hayden_black_2007}, the spreading of quantum information~\cite{sekino_fast_2008}, and principles of quantum thermodynamics~\cite{Munson2025Mar}. With such broad relevance, both practical and foundational, a natural question is how randomness arises in quantum theory and how it can be used efficiently.

The most direct construction of quantum randomness is to sample unitary matrices uniformly from the Haar measure over the unitary group on
$n$ qubits. While mathematically elegant, this approach quickly becomes unfeasible: generating Haar-random unitaries requires circuits of exponential depth in $n$, a fact already evident from basic dimensional counting~\cite{knill1995approximationquantumcircuits}. Fortunately, many applications do not require ``full Haar randomness''; it is often sufficient to match only the first few moments, say, up to order $k$. A more efficient approach is therefore to use ensembles of unitaries that reproduce the first $k$ moments of the Haar statistics, known as \emph{unitary $k$-designs}~\cite{emerson_pseudorandom_2003,Gross_2007}. An even weaker requirement, motivated from complexity theory, is computational pseudorandomness: \textit{pseudorandom unitaries} are ensembles of operators that cannot be distinguished from Haar random unitaries by any \textit{efficient} quantum measurement~\cite{Ji_2018,Haug_2025}. Both notions have become powerful tools in their own right and have been developed into a vibrant field of research~\cite{Ji_2018, metger2024pseudorandomunitariesnonadaptivesecurity, schuster2025randomunitariesextremelylow, ma2024constructrandomunitaries, haug2024pseudorandomunitariesrealsparse}. Strikingly, very recent advances~\cite{schuster2025randomunitariesextremelylow,laracuente2024approximateunitarykdesignsshallow} show that unitary 
$k$-designs and pseudorandom unitaries can now be implemented in circuits of logarithmic depth in the system size, a dramatic improvement over the exponential-depth constructions required for Haar-random unitaries.

Even though extremely shallow circuits can, in principle, generate ensembles of random unitaries, the true cost in many quantum hardware platforms—particularly early fault-tolerant architectures—depends not only on the \textit{number} of gates but also on the \textit{type} of gates. A central limitation arises from the \emph{Eastin–Knill theorem}~\cite{Eastin_2009}, which establishes that no quantum error correcting code can implement a universal gate set transversally, i.e., without time or space overhead. As a consequence, fault-tolerant gate-based architectures (based on a single code) are necessarily limited to schemes that require the implementation of at least one non-transversal gate, which then constitutes the most expensive component of the architecture. A widely used approach is to treat Clifford operations as having negligible cost~\cite{gottesman_stabilizer_1997} and to supplement them with a non-Clifford operation. The latter, implemented fault-tolerantly via \emph{magic-state injection}~\cite{bravyi_magicstate_2012}, extends the Clifford group to a universal gate set. Within this framework, the focus naturally shifts from reducing circuit depth to minimizing the number of non-Clifford gates required to realize unitary $k$-designs (see \cref{fig:placeholder}), since the dominant cost typically arises from the preparation of \textit{high-fidelity} magic-states. 
Indeed, despite remarkable progress in recent years~\cite{krishna2019towards,fang2024surpassingfundamentallimitsdistillation,nguyen2025good,golowich2025asymptotically,wills2025constant}, \textit{magic-state distillation} remains the most resource-intensive component of fault-tolerant quantum computation. 

\begin{figure}
    \centering
    \includegraphics[width=\linewidth]{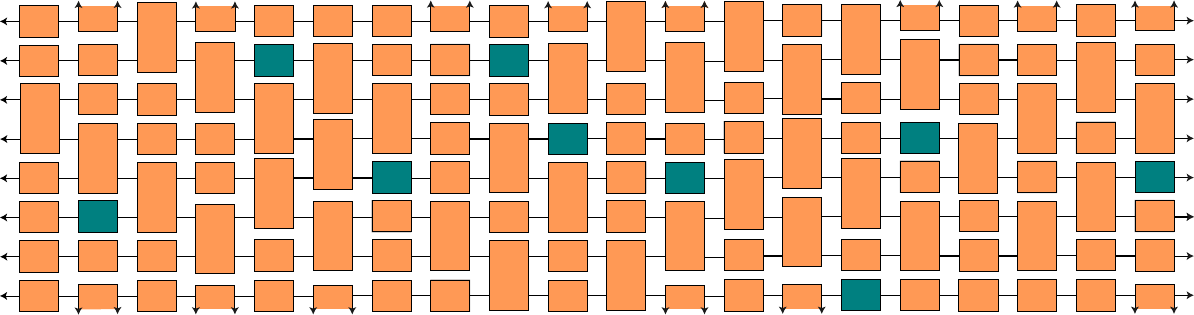}
    \caption{Clifford circuit (orange boxes) with few interspersed non-Clifford gates (green squares). }
    \label{fig:placeholder}
\end{figure}

Because of the practical relevance of this problem, substantial effort has been devoted to understanding how to construct random unitaries with minimal \textit{non-Clifford cost}. This line of work has produced a number of important results, summarized in \cref{table1}, and reviewed in Appendix S.1~\cite{_see_}. Most notably, a seminal paper~\cite{haferkamp2020QuantumHomeopathyWorks} established the existence of a scheme capable of generating unitary $k$-designs with a system-size independent number of non-Clifford gates, namely $\tilde{O}(k^4)$~\footnote{The notation $\tilde{O}$ denotes $\tilde{O}(f(n))=O(f(n)\operatorname{poly}\log f(n))$, i.e. the usual big-$O$ (or Landau) notation up to poly-logarithmic factors.}. This was a striking result: for moderate values of $k$, it implied that random unitary operators could, in principle, be implemented on near-term fault-tolerant quantum platforms with only a constant (system-size independent) resource cost. 

However, this excitement was soon tempered. The notion of unitary $k$-designs used in that result --\textit{additive-error} unitary $k$-designs-- only guarantees indistinguishability from Haar-random unitaries against a restricted class of quantum measurements, namely \textit{non-adaptive} quantum algorithms. Meanwhile, the stronger notion of \textit{relative-error} unitary designs gained popularity and is now widely regarded as the gold standard~\cite{schuster2025randomunitariesextremelylow}. Relative-error designs ensure closeness in expectation values under relative error, thereby guaranteeing indistinguishability against the most general (adaptive) quantum measurement. Unfortunately, a later result~\cite{leone2025noncliffordcostrandomunitaries} demonstrated that constructing relative-error $k$-designs necessarily requires an extensive non-Clifford cost of $\Omega(nk)$. This result appeared to rule out the efficient realization of ensembles that are operationally indistinguishable from Haar-random unitaries in fully adaptive quantum experiments. But did it, really?

\begin{table}[h!]
  \centering
  \begin{tabular}{@{}|l|c|c|@{}}
    \hline
    & \textbf{ non-Clifford cost} & \textbf{Security} \\
    \hline
    Additive $k$-design~\cite{haferkamp2020QuantumHomeopathyWorks,zhang2025designsmagicaugmentedcliffordcircuits} &    \cellcolor{lightgreen} \textcolor{black}{$O(k^4+k\log\varepsilon^{-1})$}    & \cellcolor{lightred}\textcolor{Maroon}{\xmark} \\
   Relative $k$-designs~\cite{haferkamp2020QuantumHomeopathyWorks,leone2025noncliffordcostrandomunitaries}&
   \cellcolor{lightred}\textcolor{black}{$\Theta(nk)$} & \cellcolor{lightgreen}\textcolor{ForestGreen}{\cmark} \\
  Pseudorandom unitaries~\cite{leone2025noncliffordcostrandomunitaries} &  \cellcolor{lightred}\textcolor{black}{$\Theta(n)$} &
    \cellcolor{lightblue}\textcolor{RoyalBlue}{\cmark} \\
    Secure-diluted $k$-design~[$\star$] & \cellcolor{lightgreen}\textcolor{black}{${\Theta}(k+\log\varepsilon^{-1})$\textcolor{black}{[$\diamond$]}} & \cellcolor{lightgreen}\textcolor{ForestGreen}{\cmark} \\
    Quantum-secure $k$-design~[$\star$] & \cellcolor{lightgreen}\textcolor{black}{$\tilde{O}(k^2\log \varepsilon^{-1})$} & \cellcolor{lightgreen}\textcolor{ForestGreen}{\cmark} \\
    Poly-secure $k$-design~[$\star$] & \cellcolor{lightgreen}\textcolor{black}{$\tilde{\Theta}(k+\log^{1+c}\varepsilon^{-1})$} & \cellcolor{lightblue}\textcolor{RoyalBlue}{\cmark} \\
    \hline
  \end{tabular}
    \caption{Summary of the results on unitary design constructions and their non-Clifford cost. 
Here, \xmark\ indicates that the construction is not fully secure, while \cmark\ denotes full security: 
green indicates security against all quantum experiments, and blue indicates security against 
polynomial-time quantum experiments. Quantum-secure and quantum-polynomially-secure designs are 
defined in \cref{def1,quantumpolysecuredesigns}, respectively. 
The symbol [$\star$] marks the results proven in this work, while [$\diamond$] denotes the non-Clifford cost expressed not in terms of gate count, but in terms of the support size of the diluted $k$-design in our construction (see \cref{fig1} and \cref{maintheorem}).
}\label{table1}
\end{table}


In this work, we introduce a protocol which dilutes a unitary $k$-design on $\Theta(k)$ qubits to generate unitary $k$-designs on $n$ qubits using only Clifford operations (see \cref{fig1}). The resulting designs are secure against \textit{any} quantum measurement while requiring, by construction, only a system-size-independent number of non-Clifford gates
. In addition, the construction can be implemented in logarithmic depth with the use of extra qubits. This result, which we explain in detail below, has two key implications:



\begin{enumerate}[label=(\roman*)]
    \item \textit{Practical}: it enables the implementation of random unitaries on near-term fault-tolerant quantum platforms, since the construction requires \textit{shallow} unitaries with \textit{few} non-Clifford gates;  
    \item \textit{Conceptual}: it shows that the notion of relative-error designs is not only unnecessary for physically motivated applications, but also overly demanding, requiring an extensive number of non-Clifford gates compared to the constant amount sufficient for security against \emph{any} quantum measurement.

\end{enumerate}

\textit{Setup and notation.---} We consider a system of $n$ qubits with Hilbert space $\mathbb{C}^{2\otimes n}$ and operator basis given by the Pauli group $\mathbb{P}_n$. We denote by $\|\cdot\|_1$ and $\|\cdot\|_{\infty}$ the trace norm and the operator norm of operators and by $\|\cdot\|_\diamond$ the diamond norm for quantum channels. For an ensemble $\mathcal{E}$ on the unitary group $\mathcal{U}_n$, $U \sim \mathcal{E}$ indicates that $U$ is drawn uniformly at random from $\mathcal{E}$; in particular, $\haar$ denotes the Haar (uniform) measure. We will mostly focus on the Clifford group $\mathcal{C}_n$, the subgroup of unitaries that maps Pauli operators to Pauli operators.

\medskip

\textit{Unitary $k$-designs.---} We begin with an overview of the rich landscape of unitary $k$-designs, unraveling the various concepts and nuances behind them. Let $\mathcal{E}$ be an ensemble of unitary operators on $n$ qubits, and define its $k$-moment operator as $\Phi_{\mathcal{E}}(\cdot)\coloneqq\mathbb{E}_{U\sim\mathcal{E}}U^{\otimes k}(\cdot)U^{\dag\otimes k}$. The ensemble $\mathcal{E}$ forms an {\em additive error $\varepsilon$-approximate unitary $k$-design} if~\cite{Harrow2009Random2-designs,brandao_local_2016}:
\begin{align}\label{eq:additivedesign}
    \|\Phi_{\mathcal{E}}-\Phi_{\haar}\|_{\diamond}\le \varepsilon\,,
\end{align}
where $\|\cdot\|_{\diamond}$ denotes the diamond norm. Using its variational definition, i.e. $\max_{\rho}\|I\otimes \Phi_{\mathcal{E}}(\rho)-I\otimes \Phi_{\haar}(\rho)\|_1$, one can deduce its operational meaning: given a quantum state $\rho$ with arbitrary additional qubits $n'$, an $\varepsilon$-approximate $k$-design in additive error cannot be distinguished from the Haar distribution by any quantum measurement, with resolution greater than $\varepsilon$, making $k$ \textit{parallel queries} to the unitary. This, however, does not cover the most general scenario: the most general quantum measurement may \textit{adaptively} choose subsequent operations based on intermediate measurement outcomes.

The notion of {\em relative error $\varepsilon$-approximate unitary $k$-designs} is instead defined by the closeness of the quantum channels $\Phi_{\mathcal{E}}$ and $\Phi_{\haar}$ in relative error~\cite{brandao_local_2016}:
\begin{align}\label{eq:relative designs}
    \Phi_{\haar}(1-\varepsilon)\le \Phi_{\mathcal{E}}\le (1+\varepsilon)\Phi_{\haar}\,,
\end{align}
where $\le$ denotes the ordering of completely positive maps, i.e. $\Phi\le\Phi'$ if and only if $(\Phi'-\Phi)$ is completely positive. From \cref{eq:relative designs}, one deduces the operational meaning: for any positive $(n+n')$-qubit matrix $\tilde{\rho}$, the expectation value of any operator $O$ differs by at most $2\varepsilon$ in relative error when either $I\otimes \Phi_{\mathcal{E}}$ or $I\otimes \Phi_{\haar}$ is applied. Therefore, \cref{eq:relative designs} implies \cref{eq:additivedesign}, but is strictly stronger. In fact, it can be shown~\cite{schuster2025randomunitariesextremelylow} that the state produced by any quantum algorithm making $k$ arbitrary queries to $U$ can be written as the partial overlap between one unnormalized state obtained by $k$ parallel queries to $U$ and another unnormalized state. Since relative error is insensitive to normalization, the outputs of any measurement when $U$ is sampled from $\mathcal{E}$ or from Haar differ by at most $2\varepsilon$. This establishes that relative error $k$-designs are secure against all kind of quantum measurements.  


Thus, using relative error designs that require $\Theta(nk)$ non-Clifford gates~\cite{leone2025noncliffordcostrandomunitaries}, the best known non-Clifford cost scaling to guarantee security against general quantum measurements is $O(nk)$. In contrast, as anticipated, our construction yields an $\varepsilon$-approximate unitary $k$-design secure against all quantum measurement with a system-size independent number of non-Clifford gates. 
The key step is to move beyond relative designs and instead focus exclusively on the strictly weaker—yet fully sufficient for all physically motivated purposes—requirement of security against the most general quantum algorithm. 

To formalize this, we consider the most general scenario where an adversary has query access to $k$ copies of an $n$-qubit unitary $U$. The adversary may operate on $n+n'$ qubits (including auxiliary workspace) and adaptively choose operations depending on intermediate measurement outcomes.  

We capture this idea with the following definition of a \textit{quantum-secure unitary $k$-design}~\footnote{A very recent work~\cite{cui2025unitarydesignsnearlyoptimal} introduced the notion of quantum-secure designs (referred to as \textit{measurable-error designs}), with the primary goal of minimizing circuit depth. However, with respect to gate count, their results offer no asymptotic improvement over relative-error designs, as both require $\Theta(nk)$ gates. In sharp contrast, our results show that when it comes to non-Clifford cost, relative-error designs impose an unnecessary burden: they require $\Omega(n)$ non-Clifford gates, whereas the practically motivated {quantum-secure designs} of \cref{def1} need only $O(1)$.}.  
\begin{definition}[Quantum-secure unitary designs]\label{def1} 
Let $\mathcal{E}$ be an ensemble of $n$-qubit unitaries and let $\haar$ denote the Haar measure on $n$ qubits. Let $\boldsymbol{V}$ be a collection of $k$ arbitrary unitary operators acting on $n+n'$ qubits, and define $\ket{\Psi_U(\boldsymbol{V})}\coloneqq U V_{k}UV_{k-1}\cdots UV_1\ket{0}$. $\mathcal{E}$ is an $\varepsilon$-approximate quantum secure unitary $k$-design iff:
\begin{align}
    \sup_{n'}\max_{\boldsymbol{V}}\left\|\underset{U\sim\mathcal{E}}{\mathbb{E}}\Psi_U(\boldsymbol{V})-\underset{U\sim\haar}{\mathbb{E}}\Psi_U(\boldsymbol{V})\right\|_1\le \varepsilon.
\end{align}
\end{definition}

Informally, the above norm measures the maximum distinguishing advantage as perceived by any quantum experiment making at most $k$ queries to a unitary $U$ sampled either from $\mathcal{E}$ or from the Haar measure. By construction, quantum-secure designs imply additive error designs and, as explained above, are implied by relative error designs:
\begin{align}
\text{relative-error} \mathrel{\substack{\Longrightarrow \\ \centernot\Longleftarrow}} \text{quantum-secure} \mathrel{\substack{\Longrightarrow \\ \centernot\Longleftarrow}} \text{additive-error}\label{eq:chainofimplications}\,.
\end{align}

Pseudorandom unitaries differ conceptually from unitary $k$-designs~\cite{ji_pseudorandom_2018,haug2024pseudorandomunitariesrealsparse,ma2025constructrandomunitaries}. An ensemble $\mathcal{E}$ is pseudorandom if no \textit{efficient quantum experiment} can distinguish it, up to polynomially small resolution, from Haar-random unitaries. Here, efficient experiments are those that use at most polynomially many queries $k$, with polynomially many auxiliary qubits $n'$, and implement unitaries $\boldsymbol{V}$ and POVMs of polynomial complexity in the total system size $n$. This makes pseudorandom unitaries the most practically motivated notion of randomness generation. Yet, known constructions of pseudorandom unitaries are \textit{conditional} on computational assumptions~\cite{ma2025constructrandomunitaries}. 

As in the case of relative-error designs, Ref.~\cite{leone2025noncliffordcostrandomunitaries} shows that generating pseudorandom unitaries requires $\Omega(n)$ non-Clifford gates. In practice, however, many applications only demand indistinguishability up to the first $k$ moments of the Haar distribution. The following definition combines the concepts of unitary $k$-designs and pseudorandom unitaries into what we refer to as \textit{quantum-polynomially-secure unitary $k$-designs}.  
\begin{definition}[Quantum-polynomially-secure unitary $k$-designs]\label{quantumpolysecuredesigns} 
Let $\boldsymbol{V}$ be a collection of $k$ arbitrary unitary operators acting on $n+n'$ qubits, and define $\ket{\Psi_U(\boldsymbol{V})}\coloneqq U V_{k}UV_{k-1}\cdots UV_1\ket{0}$. Let $O$ be any operator such that $\|O\|_{\infty}\le 1$. The ensemble $\mathcal{E}$ forms a quantum-polynomially-secure $\varepsilon$-approximate unitary $k$-design if 
\begin{align}
    \sup_{n',\boldsymbol{V},O\in\operatorname{poly}(n)}\left|\underset{U\sim\mathcal{E}}{\mathbb{E}}\tr[O\Psi_U(\boldsymbol{V})]-\underset{U\sim\haar}{\mathbb{E}}\tr[O\Psi_U(\boldsymbol{V})]\right|\le\varepsilon.
\end{align}
where $\operatorname{poly}(n)$ denotes the class of objects that can be implemented with resources polynomial in the total number of qubits $n$.
\end{definition}

\cref{quantumpolysecuredesigns} is similar in spirit to \cref{def1}, but it requires security only against polynomial-time quantum measurements that make at most $k$ queries to $U$, making it more relevant for practical applications. 

Before concluding this section, we note for the interested reader that the Pauli group ($k=1$) and the Clifford group ($k=2,3$) provide (trivial) examples of quantum-secure designs, as they are also exact unitary designs. The first non-trivial examples of quantum-secure designs are presented in the next section.


\medskip

{\em Our construction: quantum secure designs with constant non-Clifford cost.---} 
The idea behind our construction is simple: instead of implementing a unitary $k$-design directly on all $n$ qubits —which would incur an extensive non-Clifford cost— we apply it only to a subsystem of size $\Theta(k)$, independent of $n$. This “seed” design is then “diluted” into the full $n$-qubit space by sandwiching it between two random Clifford operators.
The resulting global operation forms an approximate quantum-secure unitary $k$-design on $n$ qubits (though not a relative-error design), while the non-Clifford cost remains confined to $\Theta(k)$ qubits. 

\begin{figure}
    \centering
    \includegraphics[width=\linewidth]{compression.pdf}
    \caption{Our construction: mixing power of Clifford circuits ``dilute'' the design ``seed'' into the full $n$-qubit space. The figure provides a graphical representation of \cref{maintheorem}.}
    \label{fig1}
\end{figure}

Concretely, we consider the ensemble of unitaries $\mathcal{E}_t\coloneqq\{C_1U_tC_2\}$ where $C_1 ,C_2\sim\mathcal{C}_n$ are uniformly sampled from the Clifford group $\mathcal{C}_n$ and $U_t$ from the Haar measure, or from an \textit{exact} $k$-design, on $t$ qubits, see \cref{fig1}. It is worth noting that the Clifford group alone already generates exact unitary 3-designs~\cite{webb_clifford_2016,zhu_multiqubit_2017}, which implies that Clifford unitaries are powerful scramblers of information, the property that is primarily responsible for our construction to work.
The next theorem, which constitutes the main result of this work, employs state-of-the-art techniques on the \textit{commutant} of the Clifford group~\cite{bittel2025completetheorycliffordcommutant,bittel2025operationalinterpretationstabilizerentropy} to show that this ensemble forms an $\varepsilon$-approximate quantum-secure unitary $k$-design.

\begin{theorem}\label{maintheorem}
    The ensemble of unitaries $\mathcal{E}_t$ forms a $\varepsilon$-approximate quantum-secure unitary $k$-design provided that $t\ge 2k+\log\varepsilon^{-1}+6$.
\end{theorem}
\noindent
\textit{Proof sketch.} We derive an explicit tight bound on the norm defined in \cref{def1}. The key step in the proof relies on a functional analysis lemma of independent interest. Specifically, consider a function of theform $f(x) = \sum_{i=1}^{k} c_i \, 2^{-i x}$. If \(|f(x)| \leq 1\) for every integer \(x \leq k\), then the coefficients satisfy\(|c_i| \leq 2^{k i}\). The full proof is deferred to Appendix S.3~1~\cite{_see_}.\qed

The above result shows that the \textit{dilution} of a $k$-design across $O(k)$ qubits through the application of random Clifford operators suffice to spread the property around the macroscopic system of $n$ qubits. Moreover, we show that $\Omega(k)$ is the minimal support required to dilute an exact $k$-design into a quantum-secure design on $n$ qubits, thereby establishing the optimality of \cref{maintheorem}.

An immediate question is how our construction translates into bounds on the number of non-Clifford gates required for quantum-secure unitary $k$-designs. Since the only step involving non-Clifford operations is the injected $k$-design on $\Theta(k)$ qubits, the resulting non-Clifford cost is \textit{automatically system-size independent}, establishing the central result of this work. The exact number of non-Clifford gates depends on the best available constructions of unitary $k$-designs. As an example, in the following corollary we apply the result of Ref.~\cite{schuster2025randomunitariesextremelylow} and use a hybrid argument, replacing the exact $k$-design on $\Theta(k)$ qubits with a relative-error approximate $k$-design (see \cref{eq:relative designs}), to derive an explicit bound on the non-Clifford gate count \textit{sufficient} for our construction.

\begin{corollary}\label{cor1}
    Quantum-secure unitary $k$-designs can be constructed using $O(k^2\log^2\varepsilon^{-1})$ many non-Clifford gates~\footnote{The dependence on the resolution error $\varepsilon$ can be improved to $\tilde{O}(k^2\log\varepsilon^{-1})$ by using $\tilde{O}(k^2)$ extra qubits, using the construction of Ref.~\cite{cui2025unitarydesignsnearlyoptimal}; see the supplemental material for details.}.
\end{corollary} 

The proof has to be found in Appendix S.3~3~\cite{_see_}. A natural question for the curious reader is how many non-Clifford gates are \textit{necessary} to construct quantum-secure unitary designs, and whether \cref{cor1} achieves optimality in this regard. The following theorem establishes a lower bound on the non-Clifford cost required for unitary designs. Specifically, we present a \textit{computationally efficient} quantum algorithm which, given $k$ parallel queries to a unitary sampled either from an ensemble $\mathcal{E}$ or from the Haar measure, can distinguish the two unless the unitaries in $\mathcal{E}$ contain at least $\Omega(k)$ non-Clifford gates.

\begin{theorem}\label{th:lowerboundmain}
    Let $\mathcal{E}$ be an ensemble of unitaries containing at most $t<n$ non-Clifford gates. If $12t+10\le k$, then
    \begin{align}
        \|\Phi_{\mathcal{E}_t}-\Phi_{\haar}\|_{\diamond}\ge\frac{1}{8}
    \end{align}
\end{theorem}
\noindent
\textit{Proof sketch:} The distinguishing algorithm proceeds as follows. We construct $k$ copies of $\ket{\psi}=U\ket{0}$. We demonstrate that $O(k)$ measurements on $\ket{\psi}$ are able to decide whether: Case A) $\ket{\psi}$ contains at least one nontrivial Pauli operator $P$ such that $P\ket{\psi}=\ket{\psi}$ or, Case B), there is no such Pauli operator. We then notice that for $U$ being Haar random there is no such Pauli operator with probability $\sim1$, while for $U$ constructed by $O(k)$ Clifford gates that are at least $2^{n-O(k)}$~\cite{leone2023learning}. The full proof is deferred to Appendix S.3~2~\cite{_see_}. \qed

Although \cref{maintheorem} is {\em optimal} in terms of the support of the diluted $k$-design, \cref{th:lowerboundmain} shows that the minimal number of non-Clifford gates required to construct a quantum-secure design is $\Omega(k)$. This suggests that the non-Clifford gate count in \cref{cor1} may be loose by a quadratic factor. However, we note that the lower bound is derived from the existence of an \textit{ efficient} quantum algorithm that can distinguish the two ensembles, whereas \cref{cor1} guarantees security even against unbounded measurements. We accordingly defined quantum-polynomially-secure $\varepsilon$-approximate unitary $k$-designs (\cref{quantumpolysecuredesigns}) as those that cannot be distinguished from Haar-random unitaries by any $k$-query computationally efficient quantum algorithm with resolution greater than $\varepsilon$. The following corollary establishes, up to computational assumptions detailed in Appendix S.3~3~\cite{_see_}, the optimality of our construction for polynomially-secure designs.
\begin{corollary}\label{cor2}
    Let $k=\omega(\log n)$, $\varepsilon=\exp(o(n))$. For any $c>0$, quantum-polynomially-secure  unitary $k$-design can be constructed using $\tilde{O}((k+\log^{1+c}\varepsilon^{-1}))$ non-Clifford gates. Combined with \cref{th:lowerboundmain}, this construction is optimal. 
\end{corollary}

This result establishes that constructing unitaries that reproduce the first $k$ moments of Haar-random unitaries for any efficiently implementable quantum measurement requires—and suffices with—only $\tilde{\Theta}(k)$ non-Clifford gates, probing the fundamental limits for unitary designs with minimal non-Clifford resources. 
As detailed in Appendix S.3~3~\cite{_see_}, the key idea in proving \cref{cor2} is to replace the exact $k$-design in our construction with a pseudorandom unitary, which functions as a quantum-polynomially secure design. However, under standard computational assumptions~\cite{ma2025constructrandomunitaries,schuster2025randomunitariesextremelylow}, such a pseudorandom unitary is secure against polynomial-time quantum adversaries only when implemented on $\omega(\log n)$ qubits. This requirement explains the assumption on $k$ in \cref{cor2}.

As a final remark before concluding this section, we note that our construction is nontrivial only when $k \le n$. However, this is precisely the regime of interest for constructing designs with a system-size–independent number of non-Clifford gates. Indeed, by \cref{th:lowerboundmain}, any unitary $k$-design with $k \ge n$ necessarily requires $\Omega(n)$ non-Clifford gates. In this case, one can instead employ constructions of random unitaries in shallow depth~\cite{schuster2025randomunitariesextremelylow,cui2025unitarydesignsnearlyoptimal} to minimize the use of non-Clifford gates, achieving $O(nk)$ gates for quantum-secure $k$-designs and $O(n)$ gates for quantum-polynomially secure designs for any $k = O(\poly n)$~\cite{leone2025noncliffordcostrandomunitaries}.


{\em Low depth and few non-Clifford gates.---}
Our construction realizes a quantum-unitary $k$-design with very few non-Clifford gates, representing an important step toward the practical implementation of random unitaries on fault-tolerant platforms. Reducing the number of non-Clifford gates is essential because their logical implementation requires either deep physical circuits or complex error-detection schemes. Without such a reduction, the shallow constructions of Ref.~\cite{schuster2025randomunitariesextremelylow} remain shallow only at the logical level—the depth at the physical level increases significantly once non-Clifford operations are made fault-tolerant. That said, besides limiting non-Clifford gates, it is also desirable to keep the overall depth shallow; otherwise, the total number of physical operations still becomes too large. Our construction can achieve both goals: building on previous results~\cite{moore1998parallelquantumcomputationquantum,jiang2022optimalspacedepthtradeoffcnot}, we can generate the two random Clifford circuits in logarithmic depth while keeping the number of non-Clifford gates low, as formalized in the next corollary and explictly proved in Appendix S.3~3~\cite{_see_}.
\begin{corollary}\label{cor3}
Quantum-secure unitary $k$-designs on $n$ qubits can be realized with $O(k^2 \log \varepsilon^{-1})$ non-Clifford gates, $O(n^2 + k^2)$ extra qubits, and an overall depth $O(\log n + \log k)$ under all-to-all connectivity.
\end{corollary}
As a final remark, the use of extra qubits and all-to-all connectivity seem to be necessary to achieve a shallow-depth version of our construction, as Ref.~\cite{grevink2025glueshortdepthdesignsunitary} shows a linear lower bound for implementing a Clifford $k$-designs with $k \ge 4$.

{\em Outlook.---} Our results show that higher-order unitary designs can be implemented with a system-size independent number of non-Clifford gates and in low depth, while remaining indistinguishable from exact designs by any quantum protocol. Unlike previous works, which either only guarantee security against non-adaptive adversaries~\cite{haferkamp2020QuantumHomeopathyWorks} or rely on relative-error designs requiring $\Omega(nk)$ non-Clifford gates~\cite{leone2025noncliffordcostrandomunitaries}, our construction achieves \textit{full quantum security} with a system-size independent and at a substantially reduced non-Clifford cost: $\tilde{O}(k^2)$ non-Clifford gates, compared to the $\tilde{O}(k^4)$ scaling of the original proposal~\cite{haferkamp2020QuantumHomeopathyWorks}. Moreover, according to \cref{th:lowerboundmain,cor2}, our construction also employs an \textit{optimal} number $\tilde{\Theta}(k)$ of non-Clifford gates if one is only concerned with security against polynomial-time quantum adversaries. 

We therefore deduce two important implications: on the conceptual level, this results provide advances in the theory of unitary $k$-designs by showing that quantum-secure designs embody the minimal and most natural notion of security whereas relative-error designs impose unnecessarily strong requirements. On the practical level, our construction is well suited for near-term implementation on early fault-tolerant devices for two main reasons: (1) all non-Clifford resources are concentrated on a constant, system-size–independent set of qubits, and (2) the overall circuit can be implemented in shallow depth.

Finally, we remark that to dilute an exact $k$-design into a quantum-secure unitary $k$-design via Clifford operations, it is both necessary and sufficient to place the $k$-design on $\Theta(k)$ qubits. This makes our construction automatically implementable with a system-size-independent number of non-Clifford gates. While this is optimal in terms of support, there is still room for improvement in the state-of-the-art constructions of (polynomially) quantum-secure designs on $\Theta(k)$ qubits. Combining our approach with the latest unitary designs constructions could further tighten the results of \cref{cor1,cor2}, improving the exact count and scaling of non-Clifford gates, as well as extending their regime of validity.

{\em Acknowledgements.---} The authors thank Salvatore F.E. Oliviero for creating the images of the paper, and Armanda Quintavalle for important discussions about Clifford+T fault-tolerant architectures. This work has been supported by the DFG (CRC 183, FOR 2724), by the BMBF (Hybrid++, QuSol), the BMWK (EniQmA), the Munich Quantum Valley (K-8), the QuantERA (HQCC),
the Alexander-von-Humboldt Foundation, Berlin Quantum and the European Research Council (ERC AdG DebuQC). This work has also been funded by the DFG under Germany's Excellence Strategy – The Berlin Mathematics Research Center MATH+ (EXC-2046/1, project ID: 390685689). L.L. acknowledges funding from the Italian Ministry of University and Research, PRIN PNRR 2022, project “Harnessing topological phases for quantum technologies”, code P202253RLY, CUP D53D23016250001, and PNRR-NQSTI project ”ECoN: End-to-end long-distance entanglement in quantum networks”, CUP J13C22000680006.

\let\oldaddcontentsline\addcontentsline
\renewcommand{\addcontentsline}[3]{}

\let\addcontentsline\oldaddcontentsline
\appendix
\onecolumngrid
\clearpage
\begin{center}
    {\normalfont\Large\bfseries Supplemental material: Adaptive Quantum Homeopathy}
\end{center}
\setcounter{secnumdepth}{2}
\setcounter{equation}{0}
\setcounter{figure}{0}
\setcounter{table}{0}
\setcounter{section}{0}
\renewcommand{\thetable}{S\arabic{table}}

\renewcommand{\thefigure}{S\arabic{figure}}
\renewcommand{\thesection}{S.\arabic{section}} 
\counterwithout{equation}{section}
\renewcommand{\theequation}{S\arabic{equation}}
\tableofcontents

\section{Overview of previous results}

In this section, we review previous results on unitary $k$-designs, with a focus on constructions that aim to minimize the number of non-Clifford gates.
Unitary $k$-designs were introduced to overcome the exponential gate cost of synthesizing Haar-random unitaries~\cite{emerson_pseudorandom_2003,emerson_scalable_2005}. For most applications, it is sufficient to reproduce the first $k$ moments of Haar statistics. The first natural approach was to construct $k$-designs using local random quantum circuits, i.e., brickwork circuits with arbitrary connectivity (1D, 2D, or all-to-all), where the two-qubit gates are sampled uniformly from the Haar measure over $\mathcal{U}_4$. This idea turned out to be correct: random quantum circuits were soon shown to form a unitary 2-design~\cite{harrow_random_2009}, and later even polynomial designs~\cite{brandao_local_2016}. These results marked a breakthrough, as they demonstrated that the randomness required for many quantum applications could be achieved in BQP, i.e., with polynomial-time quantum computation.
Despite later improvements~\cite{haferkamp_random_2022} in reducing the gate count and depth of brickwork architectures, the required depth has always been $O(\poly(n,k))$.  While still efficient in principle, such depths are far from practical in the presence of noise, as only a fully fault-tolerant quantum computer could reliably implement them. This motivated a shift in research focus toward reducing the circuit depth of $k$-designs.
A major breakthrough came in Ref.~\cite{schuster2025randomunitariesextremelylow}, which showed that random quantum circuits in a 1D architecture already form unitary $k$-designs in depth $O(k\cdot \poly\log k\cdot \log n/\varepsilon)$. This is linear in $k$ and, crucially, logarithmic in the number of qubits $n$. The key tool enabling this result is the \textit{gluing lemma} (Theorem 1 of Ref.~\cite{schuster2025randomunitariesextremelylow}), which informally states that two overlapping blocks of unitary $k$-designs form a $k$-design over their joint support. This simplified the analysis and led to a drastic reduction in required depth.
The gluing lemma turned out to be specific to the unitary group. A later work~\cite{grevink2025glueshortdepthdesignsunitary} asked whether it also holds for other groups, and obtained several no-go results: in particular, the lemma fails for the Clifford group. As a result, Clifford $k$-designs cannot be realized in logarithmic depth in $n$ without using extra qubits. However, Clifford circuits can still be implemented in depth $O(\log n)$ if $O(n^2)$ auxiliary qubits are available~\cite{moore1998parallelquantumcomputationquantum,jiang2022optimalspacedepthtradeoffcnot}. In general, auxilias can drastically reduce depth. For example, Ref.~\cite{cui2025unitarydesignsnearlyoptimal} showed that using $\tilde{O}(nk)$ extra qubits, the depth of unitary $k$-designs can be reduced to $O(\log k \log\log (kn)/\varepsilon)$, although the total gate count remains $O(nk)$, with no asymptotic improvement.
All constructions described so far (among many others not covered here) require $O(nk)$ gates in total, including non-Clifford gates. From a practical perspective, this makes them less attractive, since non-Clifford gates are expensive to implement fault-tolerantly.
The first work to show that unitary $k$-designs can be realized with few non-Clifford gates was the seminal paper~\cite{haferkamp2020QuantumHomeopathyWorks}. It proved that for $k = O(\sqrt{n})$, approximate $k$-designs with additive error can be constructed with $\tilde{O}(k^4)$ non-Clifford gates. Later, Ref.~\cite{leone2025noncliffordcostrandomunitaries,bittel2025operationalinterpretationstabilizerentropy} showed that for approximate state $k$-designs, only $O(k^2)$ non-Clifford gates are needed, provided they are supplemented with random Clifford circuits. The work~\cite{leone2025noncliffordcostrandomunitaries} also established a tight bound: $\Theta(k^2 + \log \varepsilon^{-1})$ non-Clifford gates are necessary and sufficient to match the Haar $k$-frame potential~\cite{gross_evenly_2007}. Along similar lines, \cite{p8dn-glcw} showed that, with respect to \textit{anticoncentration} (a property analogous to state $k$-designs, describing how an ensemble of states produces a statistically uniform distribution in the computational basis), Clifford circuits achieve anti-concentration at logarithmic depth. More importantly, the same work proved that only $O(\log n)$ non-Clifford gates are needed to replicate the distribution of Haar-random states. However, since the frame potential has no direct operational meaning in terms of distinguishability, we do not discuss it further. When it comes to relative-error $k$-design, the work~\cite{leone2025noncliffordcostrandomunitaries} proved a lower bound of $\Omega(nk)$ non-Clifford gates, showing that the construction in Ref.~\cite{haferkamp2020QuantumHomeopathyWorks} was optimal in that sense. As noted in the main text, this tight lower bound is the main motivation for seeking guarantees of unitary $k$-designs beyond relative designs. While relative designs ensure indistinguishability under the most general quantum measurements, they impose requirements that are too strong for practical implementation.
A concurrent work~\cite{zhang2025designsmagicaugmentedcliffordcircuits} refined these results. Reproducing the $O(k^2)$ scaling for state designs from Ref.~\cite{leone2025noncliffordcostrandomunitaries}, it also achieved a major improvement: despite the failure of the gluing lemma for Clifford unitary $k$-designs, it showed that Clifford state $k$-designs do admit gluing. This enabled a significant depth reduction: state $k$-designs can be realized in depth $O(\log n)$ for 1D architectures and $O(\log\log n)$ with all-to-all connectivity. This brings state designs with few non-Clifford gates much closer to practical use.
The same work also attempted to improve on the additive-error unitary designs of Ref.~\cite{haferkamp2020QuantumHomeopathyWorks}. Using a similar \textit{homeopathy approach}, they showed that injecting a $k$-design on $O(k^3)$ qubits suffices to construct an approximate unitary $k$-design with total non-Clifford cost $\tilde{O}(k^4)$. While this analysis does not improve on the original scaling, the key insight—that small injections of randomness can suffice—aligns with the intuition developed in the present manuscript.
One technical point deserves clarification. In most results on designs with few non-Clifford gates discussed above (except the last one), the condition $k = O(\sqrt{n})$ appears. This stems from the asymptotic behavior of Clifford Weingarten functions, proven in Refs.~\cite{bittel2025completetheorycliffordcommutant,leone2025noncliffordcostrandomunitaries}, which are proven to be diagonal in the space of Pauli monomials (\cref{def:paulimonomials}) only when $k = O(\sqrt{n})$. This explains the presence of the restriction in earlier works.

Our results improve over previous work in three ways:
\begin{itemize}
    \item {\em Stronger security}: We show that the homeopathy construction achieves adaptive security for unitary $k$-designs.
\item {\em Improved non-Clifford cost}: We prove that $\tilde{O}(k^2)$ non-Clifford gates are sufficient for quantum-secure designs (\cref{def1}). Moreover, we establish that an optimal number $\tilde{\Theta}(k)$ is both necessary (\cref{th:lowerbound}) and sufficient (\cref{cor1}) for constructing quantum-polynomially-secure designs (\cref{quantumpolysecuredesigns}).
\item {\em No artificial restrictions}: Unlike previous results, our construction is not restricted by any relation between $k$ and $n$. In particular, as long as $k<O(n)$ we place a unitary $k$-design on $\Theta(k)$ qubits and dilute it with random Clifford operations; as soon as $k=\Omega(n)$, we directly implement a unitary design on $n$ qubits.
\end{itemize}

\section{Preliminaries}\label{sec:preliminaries}
\subsection{Notation}
In this section, we briefly introduce the tools that are necessary to introduce and prove the main results in this work.
Let $\mathcal{H}$ be the Hilbert space of $n$ qubits and $d=2^n$ its dimension. We denote $\tr_{n'}(\cdot)$ the partial trace on $n'\le n$ qubits and omit the subscript whenever the trace is over the entire space.
We use the notation $[n]$ to denote the set $[n]:=\{1,\dots,n\}$, where $n\in\mathbb{N}$.
The finite field $\mathbb{F}_2$ consists of the elements $\{0, 1\}$ with addition and multiplication defined modulo 2. For $x\in\mathbb{F}_{2}^{k}$, we denote $|x|$ the Hamming weight of $x$. The space of $k \times m$ binary matrices over $\mathbb{F}_2$ is denoted by $\mathbb{F}_2^{k \times m}$. The set \( \mathrm{Sym}(\mathbb{F}_2^{m\times m}) \) denotes the set of all symmetric \( m \times m \) matrices over the finite field \( \mathbb{F}_2 \), having a null diagonal. That is,
\[
\mathrm{Sym}(\mathbb{F}_2^{m\times m}) \coloneqq \left\{ M \in \mathbb{F}_2^{m \times m} : M^T = M\,,\, M_{j,j}=0\,\,\forall j\in[m] \right\}.
\]
Similarly, we define set $\even$ as the set of binary matrices $V\in\mathbb{F}_{2}^{k\times m}$ with $m\in[k]$ with column vectors $V_{i}^{T}$ of even Hamming weight
\be
\mathrm{Even}(\mathbb{F}_{2}^{k\times m})\coloneqq\{V\in\mathbb{F}_{2}^{k\times m}\,:\, |V^{T}_{i}|=0\mod 2\,\,\forall i\in[m]\}.
\ee
\subsection{Pauli and Clifford group}\label{sec:paulioperators}
Let us introduce the qubit Pauli matrices $\{I,X,Y,Z\}$ as
\be 
I=
\begin{pmatrix}
    1 & 0 \\
    0 & 1
\end{pmatrix}\,,\quad
X=
\begin{pmatrix}
    0 & 1 \\
    1 & 0
\end{pmatrix}\,,\quad Y=
\begin{pmatrix}
    0 & -i  \\
    i &  0
\end{pmatrix}\,,
\quad Z=
\begin{pmatrix}
    1 & 0 \\
    0 & -1 
\end{pmatrix}\,.
\ee 
The Pauli group on \( n \)-qubits is defined as the \( n \)-fold tensor product of the set of Pauli matrices, multiplied by a phase factor from \( \{\pm 1, \pm i\} \). We denote the Pauli group modulo phases as \( \mathbb{P}_n \coloneqq \{I, X, Y, Z\}^{\otimes n} \).
Let us define the function 
\begin{align}\label{def:chidef}
\chi(A, B) \coloneqq \frac{1}{d} \tr(A B A^\dagger B^\dagger),
\end{align}
where \(A\) and \(B\) are operators. The function \(\chi(A, B)\) satisfies the symmetry properties:
\begin{align}
\label{eq:trivialpropCHI}
\chi(A, B) = \chi(\phi_A A, \phi_B B), \quad \chi(A, B) = \chi(B^\dagger, A) = \chi(A^\dagger, B^\dagger),
\end{align}
where $\phi_A, \phi_B \in \{-1, 1, -i, +i\}$ and the latter follows from the cyclicity of the trace. Note that, the following property holds:
\(\chi(P, Q)\) can be written as
\begin{align}
        \chi(P, Q) = 
        \begin{cases} 
        1, & \text{if } [P, Q] = 0, \\ 
        -1, & \text{if } \{P, Q\} = 0,
        \end{cases}
\end{align}
and we have that \( PQP = \chi(P, Q) Q \).

A Clifford operator is a unitary operator that maps any Pauli operator to another Pauli operator under conjugation. This family of unitaries is extremely important in quantum information theory, many-body physics~\cite{PhysRevLett.134.150403}, and it plays a dominating role in quantum error correction.

\begin{definition}[Clifford group]
    The set of Clifford operators forms a group, called the Clifford group $\mathcal{C}_n$, which is a subgroup of the unitary group $\mathcal{U}_n$, defined as the normalizer of the Pauli group $\mathbb{P}_n$.  
Concretely, this means that  
\[
\mathcal{C}_n = \{ U \in \mathcal{U}(2^n) \mid U P U^{\dagger} \in \pm\mathbb{P}_n, \, \forall P \in \mathbb{P}_n \}.
\]
It is straightforward to verify that this leads to the group property of $\mathcal{C}_n$.
\end{definition}

\begin{definition}[Stabilizer states] Stabilizer states are pure quantum states obtained by $\ket{0}^{\otimes n}$ via the action of a Clifford unitary operator $C\in\mathbb{C}_n$. The set of stabilizer states on $n$ qubits will be denoted as $\stab_n$.
\end{definition}

\subsection{Unitary $k$-designs}\label{sec:unitarykdesign}

Unitary designs are sets of unitaries that approximate Haar-random unitaries by reproducing the expectation values of low-order polynomials. In this section, we briefly introduce the relevant concepts.

We denote as $\de U$ the Haar measure over the unitary group $\mathcal{U}_n$, which is the only left/right invariant measure on $\mathcal{U}_n$. 

\begin{definition}[$k$-fold channel]\label{def:kfoldchannel} Let $O\in\mathcal{B}(\mathcal{H})$ be a operator. The $k$-fold channel (or twirling) is defined as
\be\label{eq:kfoldchannelhaar}
\Phi_{\haar}^{(k)}(O)\coloneqq\int\de U\,U^{\otimes k}OU^{\dag\otimes k}.
\ee
\end{definition}
A general recipe to compute the $k$-fold channel associated to the unitary group is via the Weingarten functions~\cite{weingarten_asymptotic_1978}.
\begin{lemma}[Weingarten calculus]\label{lem:haarweingarten} Let $S_k$ be the symmetric group on $n$ qubits and $T_{\pi}$ unitary representation of permutations $\pi\in S_k$ on $k$ tensor copies of $\mathcal{H}$. Let $\boldsymbol{\Lambda}$ the $k!\times k!$ matrix with components $\boldsymbol{\Lambda}_{\pi, \sigma}=\tr(T_{\pi}T_{\sigma})$. Let $O\in\mathcal{H}^{\otimes k}\otimes \mathcal{H}_{n'}$, where $\mathcal{H}_{n'}$ denotes a Hilbert space of arbitrary $n'$ qubits. The action of the $k$-fold twirling channel $\Phi_{\haar}\otimes I$  reads
\be
[\Phi_{\haar}^{(k)}\otimes I](O)=\sum_{\pi,\sigma}(\boldsymbol{\Lambda}^{-1})_{\pi,\sigma}T_{\sigma}\otimes\tr_{kn}(T_{\pi}^{\dag}O)
\ee
where the components $(\boldsymbol{\Lambda}^{-1})_{\pi, \sigma}$ are called Weingarten functions.
\end{lemma}

The following lemma rigorously establishes the approximation of the Haar twirl with a ``diagonal action'' of permutations.

\begin{lemma}[Approximate Haar twirl~\cite{schuster2025randomunitariesextremelylow}]\label{lem:approximatehaartwirl} Let $\Phi_a(\cdot)\coloneqq\frac{1}{d^k}\sum_{\pi\in S_{k}}\tr(T_\pi^{\dag}(\cdot))T_\pi$. It holds that $(1-\varepsilon)\Phi_a\le \Phi_{\haar}\le (1+\varepsilon)\Phi_a$ where $\Phi_{\haar}$ is the Haar twirl and $\varepsilon\le k^2/d$.
\end{lemma}

Let us introduce the concept of unitary \( k \)-designs. Denoting \( \mathcal{E} \subseteq \mathcal{U}_n \) as an ensemble of unitaries, either discrete or continuous, equipped with a measure \( \mathrm{d}\mu(U) \), one can define the \( k \)-fold channel of \( \mathcal{E} \) on \( O \in \mathcal{B}(\mathcal{H}) \) simply as
\begin{equation}\label{eq:k-foldoperator}
\Phi_{\mathcal{E}}^{(k)}(O) \coloneqq \int_{\mathcal{E}} \mathrm{d} \mu(U) \, U^{\otimes k} O U^{\dag \otimes k}\,.
\end{equation}
This operator is also commonly referred to as the $k$-th moment superoperator.
Whenever \( \Phi_{\mathcal{E}}^{(k)}(O) \) agrees with the \( k \)-fold channel on the full unitary group, we say that \( \mathcal{E} \) is (an exact) 
unitary \( k \)-design.

\begin{definition}[Exact unitary $k$-design] Let $\mathcal{E}\subseteq\mathcal{U}_n$ be an ensemble of unitaries equipped with a measure $\de\mu(U)$. $\mathcal{E}$ is an exact unitary $k$-design iff
\be
\Phi_{\mathcal{E}}^{(k)}(O)=\Phi_{\haar}^{(k)}(O)\label{unitarydesigndefinition}
\ee
for any $O\in\mathcal{B}(\mathcal{H})$. 
\end{definition}

Let us now define the concept of approximate unitary $k$-designs. As discussed in the main text, there are different notion of approximate unitary $k$-designs depending on the \textit{security} against types of quantum measurements. The first one we present is the concept of additive-error approximate unitary $k$-designs.

\begin{definition}[$\varepsilon$-approximate unitary design in additive error: non-adaptive]\label{def:unitarydesignadditiveerror} Let $\varepsilon>0$. Let $\mathcal{E}$ be an ensemble of unitaries equipped with the measure $\de\mu(U)$. $\mathcal{E}$ is an $\varepsilon$-approximate unitary design in additive error iff
\be
\|\Phi_{\mathcal{E}}-\Phi_{\haar}\|_{\diamond}\le\varepsilon
\ee
where $\Phi_{\mathcal{E}}$ (resp.\ $\Phi_{\haar}$) is the $k$-fold channel defined in \cref{def:kfoldchannel} and $\|\cdot\|_{\diamond}$ denotes the diamond norm.
\end{definition}
The notion of $\varepsilon$-approximate unitary designs in additive error has the operational meaning that the state produced by any quantum algorithm that queries $U^{\otimes k}$ when is sampled from $\mathcal{E}$ or according to the Haar measure is at most $\varepsilon$ far in trace distance. However, such quantum algorithms are not the most general quantum algorithms as adversaries can adaptively choose subsequent measurements based on previous measurement outcomes. For this reason, recently, the concept of relative error is becoming increasingly prominent in the context of the study of unitary $k$-designs \cite{low_pseudorandomness_2010,haferkamp_random_2022,schuster2025randomunitariesextremelylow}.

\begin{definition}[$\varepsilon$-approximate unitary design in relative error~\cite{low_pseudorandomness_2010,brandao_local_2016}]\label{def:unitarydesignrelativeerror} Let $\varepsilon>0$. Let $\mathcal{E}$ be an ensemble of unitaries equipped with the measure $\de\mu(U)$. $\mathcal{E}$ is an $\varepsilon$-approximate unitary design in relative error iff
\be
\label{re}
(1-\varepsilon)\Phi_{\haar}\le \Phi_{\mathcal{E}}\le (1+\varepsilon)\Phi_{\haar}
\ee
where $\Phi\le \Phi'$ denotes that $\Phi'-\Phi$ is a completely positive map.
\end{definition}
The operational meaning of relative error design in \cref{def:unitarydesignrelativeerror} is much stronger than the additive error in \cref{def:unitarydesignadditiveerror}. 
When the channel $\Phi_{\mathcal{E}}$ is applied to any positive operator $\tilde{\rho}$, the expectation value with any operator is at most $2\varepsilon$ far in \textit{relative error} from the value attained by $\Phi_{\haar}$. This notion of closeness is very powerful and implies, as explained in the main text, security against any type of quantum measurement, even adaptive quantum algorithms~\cite{schuster2025randomunitariesextremelylow}.  As a matter of fact, \cref{def:unitarydesignrelativeerror} implies \cref{def:unitarydesignadditiveerror} (up to a factor of $2$), while the converse is not true unless exponential factors are in the error:

\begin{lemma}[Additive 
vs.\ relative error. Lemma 3 in Ref.~\cite{brandao_local_2016}]\label{lem:additivetorelativeerror} Let $\varepsilon>0$. Let $\mathcal{E}$ be an ensemble of unitaries equipped with the measure $\de\mu(U)$. If $\mathcal{E}$ is $\varepsilon$-approximate $k$ design in additive error, then it is a $d^{2k}\varepsilon$-approximate $k$-design in relative error. Conversely, if $\mathcal{E}$ is an $\varepsilon$-approximate $k$-design in relative error, then $\mathcal{E}$ is an $2\varepsilon$-approximate $k$-design in additive error. 
\end{lemma}

At first glance, one might expect relative-error unitary designs to be the desired object of study. However, the main contribution of this paper is to show that relative-error designs demand unnecessarily strong guarantees and thus overshoot what is actually required for unitary designs.

Let us introduce the last notion of unitary $k$-design, i.e., quantum-secure unitary design.

\begin{definition}[Quantum secure unitary designs]\label{def1app} Let $\mathcal{E}=\{U_i\}$ be an ensemble of unitaries and let $\haar$ be the Haar measure on $n$ qubits. Let $V_{1},\ldots, V_k$ unitaries acting on $n+n'$ qubits. Denote 
\begin{align}
    \ket{\Psi_U(\boldsymbol{V})}\coloneqq(U\otimes I)V_{k}(U\otimes I)\cdots (U\otimes I)V_{1}\ket{0}^{\otimes n+n'}
\end{align}
and $\Psi_U(\boldsymbol{V})$ its density matrix. $\mathcal{E}$ is a quantum-secure unitary $k$-design iff
\begin{align}
    \sup_{n'}\max_{V_{1},\ldots, V_{k}}\left\|\mathbb{E}_{U\sim\mathcal{E}}\Psi_U(\boldsymbol{V})-\mathbb{E}_{U\sim\haar}\Psi_U(\boldsymbol{V})\right\|_1\le \varepsilon\,.
\end{align}
\end{definition}

Informally, the norm defined above measures the maximum distinguishing advantage between two channels as perceived by any quantum algorithm making at most $k$ queries to a unitary $U$ sampled either from $\mathcal{E}$ or from the Haar measure. As such, we have the following corollary, which is true by construction.
\begin{corollary}[\cite{schuster2025randomunitariesextremelylow}]\label{lem:relativeimpliesquantumsecurity} Let $\mathcal{E}$ an ensemble of unitaries on $n$ qubits. If $\mathcal{E}$ is a $\varepsilon$-approximate quantum-secure unitary $k$-designs, then $\mathcal{E}$ is a $\varepsilon$-approximate $k$-design in additive error. If $\mathcal{E}$ is a $\varepsilon$-approximate relative $k$-design in relative error, then $\mathcal{E}$ is a $2\varepsilon$-approximate quantum-secure $k$-design.
\end{corollary}

Summarizing, we have the following chain of implications for the various notion of unitary $k$-designs.
\begin{align}
\text{additive error design} \mathrel{\substack{\Longleftarrow \\ \centernot\Longrightarrow}} \text{quantum secure design} \mathrel{\substack{\Longleftarrow \\ \centernot\Longrightarrow}} \text{relative error design}\label{eq:chainofimplications}
\end{align}

We can go one step further and, instead of requiring security for unbounded quantum algorithm, we can relax the above definition, and ask security only against polynomial time quantum algorithms. Concretely, this means we require the number of auxiliary qubits $n'=O(\poly n)$, as well as all the other object involved to be efficiently preparable.

\begin{definition}[Quantum-polynomially-secure unitary designs]\label{def2app} Let $\mathcal{E}=\{U_i\}$ be an ensemble of unitaries and let $\haar$ be the Haar measure on $n$ qubits. Let $V_1,\ldots, V_k$ be a collection of polynomial complexity unitaries acting on $n+n'$ qubits with $n'=O(\poly n)$. Denote 
\begin{align}
    \ket{\Psi_U(\boldsymbol{V})}\coloneqq(U\otimes I)V_{k}(U\otimes I)\cdots (U\otimes I)V_{1}\ket{0}^{\otimes n+n'}\,,
\end{align}
and $\Psi_U(\boldsymbol{V})$ its density matrix. $\mathcal{E}$ is quantum-polynomially-secure $\varepsilon$-approximate unitary $k$-design iff
\begin{align}
    \max_{\substack{O, V_{1},\ldots, V_{k}\\ \text{efficient}\\n'=O(\poly n)}}&\Big|\mathbb{E}_{U\sim\mathcal{E}}\tr(O\Psi_U(\boldsymbol{V}))-\mathbb{E}_{U\sim\haar(n)}\tr(O\Psi_U(\boldsymbol{V}))\Big|\le \varepsilon
\end{align}
where $O$ is efficiently implementable hermitian operator with $\|O\|_{\infty}=1$.
\end{definition}

As expected, \cref{def1app} implies \cref{def2app}. However, \cref{def2app} does not guarantee non-adaptive security against unbounded adversaries and, therefore, it does not appear in the chain of implications in~\cref{eq:chainofimplications}.

From \cref{def2app} one can define \textit{psedorandom unitaries} as follows.
\begin{definition}[Pseudorandom unitaries]\label{def:pseudorandomunitaries} $\mathcal{E}$ is an ensemble of pseudorandom unitaries on $n$ qubits if $\mathcal{E}$ is a quantum-polynomially-secure $\varepsilon$-approximate $k$-design, with $\varepsilon=o(\operatorname{poly}n^{-1})$, for every $k=O(\operatorname{poly} n)$.
    
\end{definition}

\subsection{Haar average over the Clifford group}\label{sec:haaraveragecliff}
In this section, we summarize the findings of Ref.~\cite{bittel2025completetheorycliffordcommutant} regarding the average over the Clifford group, as it will be instrumental for the proof of our main result. 
Let us first define the set of \textit{reduced Pauli monomials}, which will be referred to as the set of Pauli monomials for brevity.

\begin{definition}[Pauli monomials]\label{def:paulimonomials}
Let \( k,m \in \mathbb{N} \), $V\in\mathrm{Even}(\mathbb{F}_{2}^{k\times m})$ with independent column vectors and $M\in\mathrm{Sym}(\mathbb{F}_{2}^{m\times m})$. A Pauli monomial, denoted as \( \Omega(V, M) \in \mathcal{B}(\mathcal{H}^{\otimes k}) \), is defined as
\begin{align}
\Omega(V, M) \coloneqq \frac{1}{d^m} \sum_{\boldsymbol{P} \in \mathbb{P}_n^m}  
P_1^{\otimes v_1} P_2^{\otimes v_2} \cdots P_m^{\otimes v_m}
\left( \prod_{\substack{i, j \in [m] \\ i < j}} \chi(P_i, P_j)^{M_{i,j}} \right),
\end{align}
where  $ \chi(P_i, P_j)$ is defined in \cref{def:chidef}, and we denote a string of Pauli operators in $\mathbb{P}_{n}$ as $\boldsymbol{P}\coloneqq P_1,\ldots, P_k$. We define the set of reduced Pauli monomials as
\be\label{eq:paulimonomialset}
\mathcal{P}\coloneqq\{\Omega(V,M)\,|\, V\in\even\,:\,\rank(V)=m\,,\,M\in\symf\,,\,m\in[k-1]\}\,. 
\ee
\end{definition}

\begin{lemma}[Relevant properties of Pauli monomials]\label{lem:relevantpropertiespaulimonomials} The following facts hold~\cite{bittel2025completetheorycliffordcommutant}:
\begin{itemize}
    \item The commutant of the Clifford group is spanned by $\mathcal{P}$.
    \item For $n\ge k-1$, $\mathcal{P}$ contains linearly independent operators.
    \item  $
|\mathcal{P}|=\prod_{i=0}^{k-2}(2^i+1) $, and the following bounds holds $2^{\frac{k^2-3k-1}{2}} \le \vert \operatorname{Comm}(\mathcal{C}_n) \vert \le 2^{\frac{k^2-3k+12}{2}}$.    
    \item For $\Omega\in \mathcal{P}$, then $\Omega=\omega^{\otimes n}$ (i.e., factorizes on qubits), where
    \be
\omega=\frac{1}{2^m} \sum_{\boldsymbol{P} \in \{I,X,Y,Z\}^m}  
P_1^{\otimes v_1} P_2^{\otimes v_2} \cdots P_m^{\otimes v_m}\times \left( \prod_{\substack{i, j \in [m] \\ i < j}} \chi(P_i, P_j)^{M_{i,j}} \right).
    \ee
\item $\Omega^{\dag}=\Omega^{T}$ for all $\Omega\in\mathcal{P}$.
\item $\Omega=\Omega_U\Omega_P$ where $\Omega_U\in\mathcal{P}_U$, $\Omega_P\in\mathcal{P}_P$ for all $\Omega\in\mathcal{P}$.
\item $\|\Omega\|_1=\tr(\Omega_P(V,0))\eqqcolon d^{k-m_p}$ where $\Omega_P(V,0)\in\mathcal{P}_P$, for  any $\Omega\in\mathcal{P}$.
\item $\tr(\Omega\rho^{\otimes k})\le1$ for any state $\rho$ and any $\Omega\in\mathcal{P}$.
\end{itemize}
\end{lemma}

Let us one relevant technical object of this work.
\begin{definition}[Inner product between Pauli monomials]\label{def:alpha} Let $\Omega,\Omega'$ two reduced Pauli monomials. Their inner product can be expressed as $\tr(\Omega^{\dag}\Omega')=d^{k-\alpha(\Omega,\Omega')}$ (see Ref.~\cite{bittel2025completetheorycliffordcommutant}). We are interested in the function $\alpha$, which can be expressed as
\be
\alpha(\Omega,\Omega')\coloneqq k-\log_{d}\tr(\Omega^{\dag}\Omega')\,.
\ee
\end{definition}

\begin{lemma}[Basic properties of the function $\alpha$]\label{lem:basicpropertiesalpha} Given two reduced monomials $\Omega,\Omega'$. Let $m_{p}$ (resp. $m_p'$) be the number of projective primitives in the normal form of $\Omega$ (resp. $\Omega'$), which coincide with the expression of the trace norm see \cref{lem:relevantpropertiespaulimonomials}. The function $\alpha(\Omega,\Omega')$ obeys the following properties: 
\begin{itemize}
    \item $\alpha(\Omega,\Omega')=\alpha(\Omega',\Omega)$;
    \item $\alpha(\Omega,\Omega')\ge |m_{p}-m_{p}'|$;
    \item $\alpha(\Omega,\Omega')=0$ if and only if $\Omega=\Omega'$;
    \item $\alpha(\Omega,\Omega')\ge |m-m'|$;
    \item $\alpha(\Omega,\Omega')\le k-1$;
\end{itemize}
\begin{proof}
    The first property follows directly from the definition. The second property follows from Hölder's inequality:
\begin{align}
    \tr(\Omega^{\dag}\Omega') \le \|\Omega\|_{1}\|\Omega'\|_{\infty} = d^{k-m_p} d^{m_p'} = d^{k-(m_p-m_p')}
\end{align}
which implies that $\alpha \ge (m_p - m_p')$. Applying Hölder's inequality again, but switching the roles of $\Omega$ and $\Omega'$, yields $\alpha \ge (m_p' - m_p)$, which proves the statement.  

To prove the third property, note that the leftward implication follows immediately from the fact that $\tr(\Omega^{\dag}\Omega) = d^{k}$ and therefore $\alpha(\Omega,\Omega) = 0$. For the rightward implication, we can write the product as $\Omega^{\dag}\Omega' = d^{\beta}\bar{\Omega}$, where $\bar{\Omega}$ is reduced. Thus,  
\[
\tr(\Omega^{\dag}\Omega') = d^{\beta} d^{k - m(\bar{\Omega})}.
\]  
Imposing $\alpha = 0$ gives $\beta = m(\bar{\Omega})$. Since $\beta$ counts the number of projective primitives shared in the decompositions of $\Omega$ and $\Omega'$, this means that $\bar{\Omega}$ must be a projector with $m(\bar{\Omega}) = m_p(\bar{\Omega}) = \beta$. This shows that the unitary parts of both $\Omega$ and $\Omega'$ must coincide, i.e., they multiply to the identity. Moreover, the argument also shows that the projective parts must be identical; otherwise, we would have $m(\bar{\Omega}) > \beta$. We therefore conclude that $\Omega = \Omega'$. For the last item, we make use of Lemma 44 in Ref.~\cite{bittel2025completetheorycliffordcommutant}. We have $\tr(\Omega^{\dagger}\Omega')=d^{k-(m+m')+2\beta}$  where $\beta$ is the number of linear dependencies between $\Omega$ and $\Omega'$. Given that $\beta\le \min(m,m')$, we have $\tr(\Omega\Omega')\le d^{k-|m'-m|}$. Given that $m\le k-1$, see \cref{eq:paulimonomialset}, it implies, combined with the definition of $\alpha$, the desired result. 
\end{proof}
\end{lemma}

The function $\alpha$ effectively plays the role of a \textit{distance} between Pauli monomials, as it obeys triangle inequality as shown in the following lemma.  
\begin{lemma}[Triangle inequality for $\alpha$]\label{lem:triangleinequality} Let $\Omega,\Omega'$ be two Pauli monomials. For any permutation $\pi$ it holds that
\begin{align}
    \alpha(\Omega,\Omega')\le \alpha(\Omega,\pi)+\alpha(\Omega',\pi)
\end{align}
\end{lemma}
\begin{proof}
    The statement we want to show is equivalent to that $\forall \Omega,\Omega',\pi$
    \begin{align}
        d^{k}\tr(\Omega'\Omega^{\dagger})\geq \tr(\pi\Omega^{\dagger})\tr(\pi\Omega^{\prime\dagger})
    \end{align}
   By redefining $\bar \Omega\eqqcolon\Omega\pi^{\dagger}$ and $\bar \Omega'\eqqcolon\Omega'\pi^{\dagger}$, we obtain the expression
      \begin{align}
d^{k}\tr(\bar\Omega'\pi\pi^{\dagger}\bar\Omega^{\dagger})\geq \tr(\bar\Omega^{\dagger})\tr(\bar\Omega^{\prime\dagger})
   \end{align}
   where the middle terms on the left hand side cancel out. As we want to proof the statement for all Pauli monomials and the described redefinition is bijective, it suffices to show
   \begin{align}
       d^{k}\tr(\bar\Omega'\bar\Omega^{\dagger})&\geq \tr(\bar\Omega^{\dagger})\tr(\bar\Omega^{\prime\dagger})\\
       d^{2k-m(\bar\Omega)-m(\bar\Omega')+2\beta(\bar\Omega,\bar \Omega')}&\geq d^{2k-m(\bar\Omega)-m(\bar\Omega')}
   \end{align}
   which holds as $\beta(\bar\Omega,\bar \Omega')\geq0 $ holds.
\end{proof}

The main result of the manuscript follows from averaging over Clifford operators. In the following, we denote 
\begin{align}
    \Phi_{\cl}(\cdot)\coloneqq\frac{1}{|\mathcal{C}_n|}\sum_{C\in\mathcal{C}_n}C^{\otimes k}(\cdot)C^{\dag\otimes k}
\end{align}
the twirling over the Clifford group.

\begin{lemma}[Twirling over the Clifford group~\cite{bittel2025completetheorycliffordcommutant}]\label{sec:cliffordweingartencalculus} Consider the Clifford group $\mathcal{C}_n$. Let $O\in\mathcal{H}^{\otimes k}\otimes \mathcal{H}_{n'}$, where $\mathcal{H}_{n'}$ denotes a Hilbert space of arbitrary $n'$ qubits. The action of the $k$-fold channel $\Phi_{\cl}(\cdot)\otimes I$ of the Clifford group reads:
\be
 [\Phi_{\cl}\otimes I](O)=\frac{1}{d^k}\sum_{\Omega,\Omega'\in\mathcal{P}}W_{\Omega,\Omega'} \Omega^{\prime}\otimes \tr_{k}[(\Omega^{\dagger}\otimes I)O]
\ee 
where $\mathcal{P}$ is the set of Pauli monomials in \cref{def:paulimonomials}, and $W_{\Omega\Omega'}$ are the Clifford-Weingarten functions introduced in~\cite{bittel2025completetheorycliffordcommutant}, obtained by inverting the Gram-Schmidt matrix $G_{\Omega,\Omega'}\coloneqq\frac{1}{d^k}\tr(\Omega^{\dag}\Omega')$.
\end{lemma}

We remark that we are using a different normalization compared to Ref.~\cite{bittel2025completetheorycliffordcommutant} for convenience.

\section{Proof of the main results}\label{sec:proofs}
In this section, we prove the main result of this work.
\subsection{Adaptive security: proof of \cref{maintheorem}}\label{sec:adaptivesecurity}
We consider the ensemble of unitaries $\mathcal{E}_t=\{C_1U_tC_2\}$ with $C_1,C_2$ uniformly sampled from the Clifford group and $U_t$ to be a unitary $k$-design on $t$ qubits. Our objective is also to show that this ensemble is a quantum-secure unitary $k$-design according to \cref{def1app} (\cref{def1} in the main text). First, we prove a preliminary lemma expressing $\tr(O\Psi_U(\boldsymbol{V}))$ in terms of the Clifford–Weingarten functions defined in \cref{sec:cliffordweingartencalculus}.
\begin{lemma}\label{lem8preliminary}
    Consider the ensemble of unitaries $\mathcal{E}_t$, and let $\Psi_U(\boldsymbol{V})$ be the state produced by an arbitrary adaptive quantum algorithm (see \cref{def1app}) with $k$ queries to $U\sim\mathcal{E}_t$. Let $O$ be any operator defined on $n+n'$ qubits with $\|O\|_{\infty}\le 1$. It holds that: 
    \begin{align}
        \mathbb{E}_{U\sim\mathcal{E}_t}\tr(O\Psi_U(\boldsymbol{V}))=2^{kt}\sum_{\substack{\Omega,\Omega',\Omega'',\Omega'''\in\mathcal{P}\\\pi,\pi\in S_k}}g(\Omega,\Omega''')d^{-\alpha(\Omega,\Omega')}(\Lambda^{-1})_{\pi\pi'}W_{\Omega\Omega'}W_{\Omega''\Omega'''}2^{-t[\alpha(\Omega',\pi)+\alpha(\Omega'',\pi')-\alpha(\Omega',\Omega'')]} 
    \end{align}
where $\alpha(\cdot,\cdot)$ is defined in \cref{def:alpha} and $g(\Omega,\Omega''')$ is defined later in the proof.
\begin{proof}
    Let us consider the state $\ket{\Psi_U(\boldsymbol{V})}$ defined in \cref{def1app} living into $\mathcal{H}_n\otimes \mathcal{H}_{n'}$. Let us consider an enlarged Hilbert space, namely $\mathcal{H}_n\otimes \mathcal{H}_{n'}\otimes \mathcal{H}^{\otimes k}_n\otimes\mathcal{H}^{\otimes k}_n$. Define the following unnormalized state 
    \begin{align}
        \ket{\tilde{\Psi}_I(\boldsymbol{V})}&\coloneqq\sum_{\substack{x_1,\ldots,x_k\\y_1,\ldots, y_k}}\prod_{i=1}^{k}(\ketbra{x_i}{y_i}\otimes I_{n'})V_i\ket{0}^{\otimes n+n'}\otimes \ket{x_1,\ldots, x_k}\otimes \ket{y_1,\ldots, y_k}
    \end{align}
    and the following other unnormalized state defined on the last two factors $\mathcal{H}^{\otimes k}\otimes \mathcal{H}^{\otimes k}$
    \begin{align}
\ket{\tilde{\Psi}_{\text{Bell}}}\coloneqq\sum_{z_1,\ldots, z_k}\ket{z_1,\ldots,z_k}\otimes \ket{z_1,\ldots, z_k}
    \end{align}
    With some manipulations, see also Ref.~\cite{schuster2025randomunitariesextremelylow}, it is possible to show that, for any unitary $U\in\mathcal{U}_n$, one can write  as
    \begin{align}
\ket{\Psi_{U}(\boldsymbol{V})}=\left(I_n\otimes I_{n'}\otimes \bra{\tilde{\Psi}_{\text{Bell}}}\right)I_n\otimes I_{n'}\otimes U^{\otimes k}\otimes I_n^{\otimes k}\ket{\tilde{\Psi}_I(\boldsymbol{V})}
    \end{align}
which gives mathematical rigorous grounds to the statement made in the main text that $\ket{\Psi_U(\boldsymbol{V})}$ can be regarded as the inner product of two unnormalized states. Let $O$ be a hermitian opeator on $\mathcal{H}_n\otimes\mathcal{H}_{n'}$. We therefore have
\begin{align}
    \tr(O\Psi_U(\boldsymbol{V}))=\tr(O\otimes \tilde{\Psi}_{\text{Bell}}[I\otimes U^{\otimes k}]\tilde{\Psi}_{I}(\boldsymbol{V})[I\otimes U^{\dag\otimes k}])
\end{align}
where $I$ above represents the indentity over all the spaces $U^{\otimes k}$ does not act. Let $U\sim \mathcal{E}_t$. We can write $U$ as $U=C_1U_tC_2$. We can therefore apply \cref{sec:cliffordweingartencalculus,lem:haarweingarten} to average over $C_1,C_2,U_t$ and obtain the following:
\begin{align}
    &\mathbb{E}_{U_t\sim\haar(t)}\mathbb{E}_{C_1,C_2\sim\mathcal{C}_n}\tr(O\Psi_U(\boldsymbol{V}))\\&=\mathbb{E}_{U_t\sim\haar(t)}\mathbb{E}_{C_1\sim\mathcal{C}_n}\frac{1}{d^k}\sum_{\Omega,\Omega'}W_{\Omega\Omega'}\tr(O\otimes \tilde{\Psi}_{\text{Bell}}\tr_{kn}(\tilde{\Psi}_I(\boldsymbol{V})\Omega^{\dag})\otimes (C_1U_t)^{\otimes k}\Omega'(C_1U_t)^{\dag\otimes k})\\
&=\mathbb{E}_{C_1\sim\mathcal{C}_n}\frac{1}{d^k}\sum_{\Omega,\Omega'}\sum_{\pi,\pi'}(\Lambda^{-1})_{\pi\pi'}W_{\Omega\Omega'}\tr(O\otimes \tilde{\Psi}_{\text{Bell}}\tr_{kn}(\tilde{\Psi}_I(\boldsymbol{V})\Omega^{\dag})\otimes (C_1)^{\otimes k}[T_{\pi}\otimes\tr_{kt}(T_{\pi}^{\dag}\Omega')](C_1)^{\dag\otimes k})\\
&=\frac{1}{d^{2k}}\sum_{\Omega,\Omega',\Omega'',\Omega'''}\sum_{\pi,\pi'}(\Lambda^{-1})_{\pi\pi'}W_{\Omega\Omega'}W_{\Omega''\Omega'''}\tr(O\otimes \tilde{\Psi}_{\text{Bell}}\tr_{kn}(\tilde{\Psi}_I(\boldsymbol{V})\Omega^{\dag})\otimes \Omega^{\prime\prime\prime\dag})\tr(\Omega''T_{\pi}\otimes \tr_{kt}(T_{\pi'}^{\dag}\Omega'))
\end{align}
Given the factorization properties of Pauli monomials, see \cref{lem:relevantpropertiespaulimonomials}, we can write $\tr(\Omega''T_{\pi}\otimes \tr_{kt}(T_{\pi'}^{\dag}\Omega'))=\tr_{kt}(\Omega''T_{\pi})\tr_{kt}(T_{\pi'}^{\dag}\Omega')\tr_{k(n-t)}(\Omega'\Omega'')$. We define 
\begin{align}
g(\Omega,\Omega''')\coloneqq\frac{1}{d^k}\tr(O\otimes \tilde{\Psi}_{\text{Bell}}\tr_{kn}(\tilde{\Psi}_I(\boldsymbol{V})\Omega^{\dag})\otimes \Omega^{\prime\prime\prime\dag})
\end{align}
and arrive to
\begin{align}
 &\mathbb{E}_{U_t\sim\haar(t)}\mathbb{E}_{C_1,C_2\sim\mathcal{C}_n}\tr(O\Psi_U(\boldsymbol{V}))\\&=\frac{1}{d^{k}}\sum_{\Omega,\Omega',\Omega'',\Omega'''}\sum_{\pi,\pi'}g(\Omega,\Omega''')(\Lambda^{-1})_{\pi\pi'}W_{\Omega\Omega'}W_{\Omega''\Omega'''}\tr_{kt}(\Omega''T_{\pi})\tr_{kt}(T_{\pi'}^{\dag}\Omega')\tr_{k(n-t)}(\Omega'\Omega'')  \\
 &=\frac{1}{d^{k}}\sum_{\Omega,\Omega',\Omega'',\Omega'''}\sum_{\pi,\pi'}g(\Omega,\Omega''')(\Lambda^{-1})_{\pi\pi'}W_{\Omega\Omega'}W_{\Omega''\Omega'''}\tr_{kt}(\Omega^{\prime\prime\dag}T_{\pi})\tr_{kt}(T_{\pi'}^{\dag}\Omega')\tr_{k(n-t)}(\Omega'\Omega^{\prime\prime\dag}) 
\end{align}
where we used the change of variables $\Omega^{\prime\prime}\mapsto\Omega^{\prime\prime\dagger}$. We can now use \cref{def:alpha} for the function $\alpha(\cdot,\cdot)$  to conclude. 
\end{proof}
\end{lemma}

Next, we prove a general lemma of function analysis, crucial for our proof of adaptive security.
\begin{lemma}\label{lemmalemma}
Let $f(x)=\sum_{i=1}^{k}a_i2^{-ix}$. If for every positive integer $x\le k$ it holds that $|f(x)|\le C$, then $a_{i}\le 30C\, 2^{ik}$.
\begin{proof}
    Let $j$ be an integer $1\le j\le k$. We have $f(j)=\sum_{i=1}^{k}a_i2^{-ij}$. Defining $M$ be the $k\times k$ symmetric matrix with components $M_{ij}=2^{-ij}$, we can write $a_{i}=\sum_{j=1}^{k}(M^{-1})_{ij}f(j)$ and bound $|a_{i}|\le C\sum_{j=1}^{k}|(M^{-1})_{ij}|$, where we used our hypothesis that for every $j\le k$ (with $j\in\mathbb{N}$) we have $|f(j)|\le C$. This means that, to bound $|a_i|$, we just need to bound $\sum_{j=1}^{k}|(M^{-1})_{ij}|$. We notice that $M_{ij}$ is a Vandermonde matrix, because $M_{ij}=(M_{1j})^{i}$. The components of the inverse of a Vandermonde matrix read~\cite[Excercise 40, 1.2.3]{knuth1997art}:
    \begin{align}
    (M^{-1})_{ij}=(-1)^{j-1}\sum_{\substack{1\le m_1< m_2<\cdots< m_{k-j}\le k\\m_1,\ldots, m_{k-j}\neq i}}x_{m_1}\cdots x_{m_{k-j}}\left[x_i\prod_{\substack{1\le m\le k\\m\neq i}}(x_m-x_i)\right]^{-1}\,,
    \end{align}
where we identify $x_{j}\coloneqq M_{1j}=2^{-j}$. Let us first lower bound the last product term.
\begin{align}
    2^{-i}\prod_{\substack{1\le m\le k\\m\neq i}}|2^{-m}-2^{-i}|&=2^{-i}\prod_{m=1}^{i-1}(2^{-m}-2^{-i})\prod_{m=i+1}^{k}(2^{-i}-2^{-m})\\
    &=2^{-i}\prod_{m=1}^{i-1}2^{-m}\prod_{m=1}^{i-1}(1-2^{-(i-m)})\prod_{m=i+1}^{k}2^{-i}\prod_{m=i+1}^{k}(1-2^{-(m-i)})\\
    &=2^{-i}2^{-\frac{i(i-1)}{2}}2^{-i(k-i)}\prod_{m=1}^{i-1}(1-2^{-m})\prod_{m=1}^{k-i}(1-2^{-m})\\
    &\le 2^{-\frac{i(i+1)}{2}-i(k-i)}\prod_{m=1}^{\infty}(1-2^{-m})^2
\end{align}
Let us upper bound the numerator:
\begin{align}
    \sum_{j=1}^{k}\sum_{\substack{1\le m_1< m_2<\cdots< m_{k-j}\le k\\m_1,\ldots, m_{k-j}\neq i}}2^{-m_1}\cdots 2^{-m_{k-j}}=\prod_{\substack{j=1\\j\neq i}}^{k}(1+2^{-j})-1\le  \prod_{{j=1}}^{\infty}(1+2^{-j})
\end{align}
By bounding the individual the products
\begin{align}
    \prod_{{j=1}}^{\infty}(1+2^{-j})&\leq 2.4\\
    \prod_{{j=1}}^{\infty}(1-2^{-j})&\geq 0.28
\end{align}
we obtain the desired general bound 
\begin{align}
    \sum_{j}|(M^{-1})_{ij}|\le 30 \times 2^{ki-\frac{i(i-1)}{2}}\,.
\end{align}
\end{proof}
\end{lemma}
Let us restate and proof our main result in \cref{maintheorem}.
\begin{theorem}[Restatement of \cref{maintheorem}] \label{th:maintheorem1}
The ensemble $\mathcal{E}_t$ forms a $\varepsilon$-approximate quantum secure $k$-design on $n$ qubits provided that $t\ge 2k+6+\log\varepsilon^{-1}$. 
    \begin{proof}
According to \cref{def1app}, in order to prove that $\mathcal{E}_t$ is a quantum secure design, we need to bound the following norm
\begin{align}\label{eqproof21}
    \sup_{n'}\max_{V_1,\ldots, V_k}\|\mathbb{E}_{U\sim\mathcal{E}_t}\Psi_U(\boldsymbol{V})-\mathbb{E}_{U\sim\haar}\Psi_U(\boldsymbol{V})\|_1
\end{align}
where $\Psi_{U}(\boldsymbol{V})$ is the state produced by a generic $k$ queries adaptive quantum strategy, defined in \cref{def1app}. Using the variational definition of the trace distance, i.e. $\|\rho-\sigma\|_1=\max_{O\,:\,\|O\|_{\infty}\le 1}\tr(O(\rho-\sigma))$, in the following we bound the difference $|\tr(O\mathbb{E}_{U\sim\mathcal{E}_t}\Psi_U(\boldsymbol{V}))-\tr(O\mathbb{E}_{U\sim\haar}\Psi_U(\boldsymbol{V}))|$ for any $n',O,V_{1},\ldots, V_k$ and, as a consequence, bound the norm in \cref{eqproof21}. Let $O$ be an arbitrary bounded operator on $n+n'$ qubits. According to \cref{lem8preliminary}, we can write  $\tr(O\mathbb{E}_{U\sim\mathcal{E}_t}\Psi_U(\boldsymbol{V}))$ as
\begin{align}
    \mathbb{E}_{U\sim\mathcal{E}_t}\tr(O\Psi_U(\boldsymbol{V}))=2^{kt}\sum_{\substack{\Omega,\Omega',\Omega'',\Omega'''\in\mathcal{P}\\\pi,\pi\in S_k}}g(\Omega,\Omega''')d^{-\alpha(\Omega,\Omega')}(\Lambda^{-1})_{\pi\pi'}W_{\Omega\Omega'}W_{\Omega''\Omega'''}2^{-t[\alpha(\Omega',\pi)+\alpha(\Omega'',\pi')-\alpha(\Omega',\Omega'')]} 
\end{align}
where the details have to be found in \cref{lem8preliminary}. Let us notice that for $t=n$, we have $\mathcal{E}_t=\haar$, so \cref{lem8preliminary} characterizes also the average for $U\sim\haar$. Before proceeding with bounding the term, let us make use of \cref{lem:approximatehaartwirl} to approximate the Haar average on $t$ qubits. We can therefore write 
\begin{align}
    |\mathbb{E}_{U\sim \mathcal{U}_t}\tr(O\Psi_U(\boldsymbol{V}))-f(\boldsymbol{V},O,t)|&\le \frac{2k^2}{2^t}\\
    |\mathbb{E}_{U\sim \haar}\tr(O\Psi_U(\boldsymbol{V}))-f(\boldsymbol{V},O,n)|&\le \frac{2k^2}{2^n}
\end{align}
where we denoted:
\begin{align}
    f(\boldsymbol{V},O,t)\coloneqq \sum_{\substack{\Omega,\Omega',\Omega'',\Omega'''\in\mathcal{P}\\\pi\in S_k}}g(\Omega,\Omega''')d^{-\alpha(\Omega,\Omega')}W_{\Omega\Omega'}W_{\Omega''\Omega'''}2^{-t[\alpha(\Omega',\pi)+\alpha(\Omega'',\pi)-\alpha(\Omega',\Omega'')]} \,.
\end{align}
Thanks to \cref{lem:triangleinequality}, we know that the function $0\le\alpha(\Omega',\pi)+\alpha(\Omega'',\pi)-\alpha(\Omega',\Omega'')\le 2k$. We can thus formally write the function  $f(\boldsymbol{V},O,t)=\sum_{x=0}^{2k}c_x2^{-xt}$ (isolating the only term depending on $t$) and, for any operator $O$ and collection of unitaries $\boldsymbol{V}$, write
\begin{align}
    |\mathbb{E}_{U\sim \mathcal{U}_t}\tr(O\Psi_U(\boldsymbol{V}))-\mathbb{E}_{U\sim \haar}\tr(O\Psi_U(\boldsymbol{V}))|&\le \left|\sum_{x=0}^{2k}c_x(2^{-xt}-2^{-xn})\right|+ \frac{2k^2}{2^t}+\frac{2k^2}{2^n}\\
    &\le 2\sum_{x=1}^{2k}|c_x|2^{-xt}+ \frac{2k^2}{2^t}+\frac{2k^2}{2^n}
\end{align}
Noting that $f(\boldsymbol{V},O,t)\le 1+\frac{2k^2}{2^t}$, let us use \cref{lemmalemma} to bound the norm in \cref{eqproof21}:
\begin{align}
    \sup_{n'}\max_{V_1,\ldots, V_k}\|\mathbb{E}_{U\sim\mathcal{E}_t}\Psi_U(\boldsymbol{V})-\mathbb{E}_{U\sim\haar}\Psi_U(\boldsymbol{V})\|_1&\le \sum_{x=1}^{2k}|c_x|2^{-xt}+ \frac{2k^2}{2^t}+\frac{2k^2}{2^n}\\
    &\le 30(1+2k^2 2^{-t})\sum_{x=1}^{2k}2^{-(t-2k)x}+\frac{2k^2}{2^t}+\frac{2k^2}{2^n}\\
    &\le 30\frac{3}{2}\sum_{x=1}^{\infty}2^{-(t-2k)x}+\frac{2k^2}{2^t}+\frac{2k^2}{2^n}\\
    &\le 30\frac{3}{2}2^{-(t-2k)}+\frac{2k^2}{2^t}+\frac{2k^2}{2^n}
\end{align}
where we imposed $t\ge 2k$. We used that $(1+2k^2 2^{-t})\le (1+2k^2 2^{-2k})\le\frac{3}{2}$ for every $k$. Requiring $t\le n$, we have finally have 
\begin{align}
\sup_{n'}\max_{V_1,\ldots, V_k}\|\mathbb{E}_{U\sim\mathcal{E}_t}\Psi_U(\boldsymbol{V})-\mathbb{E}_{U\sim\haar}\Psi_U(\boldsymbol{V})\|_1\le2^{-t}(452^{2k}+2^{\log_2k^2+2})\le 47\times 2^{-t}2^{2k}\,,
\end{align}
hence, we impose $t\ge 2k+6+\log\varepsilon^{-1}$, which proves the claim.
    \end{proof}
\end{theorem}

\subsection{Lower bounds on the non-Clifford cost}\label{sec:lowerbound}
In this section, we establish a tight lower bound for our quantum homeopathy construction. Specifically, we design an algorithm that, using $O(k)$ parallel queries, distinguishes whether a unitary is drawn from an ensemble of $O(k)$-compressible unitaries or from the Haar-random ensemble. This lower bound has two key implications:
\begin{enumerate}[label=(\alph*)]
    \item It shows that $\Omega(k)$ is the minimal qubit support of a $k$-design required to dilute it via random Clifford operations into a unitary $k$-design on $n$ qubits, thereby establishing the optimality of our homeopathy analysis in \cref{maintheorem}; see \cref{cor:quantumhomoptimal}.
    \item It provides a lower bound on the number of non-Clifford gates required to construct an $\varepsilon$-approximate unitary $k$-design with additive error, and therefore a quantum-secure design; see \cref{cor:lowerboundnoncliffordgates}.
\end{enumerate}
We remark that, in the main text, we discussed point (a) above following \cref{maintheorem} and stated \cref{th:lowerboundmain} as a consequence of \cref{cor:lowerboundnoncliffordgates} for clarity of presentation.

First, we need the following definitions. 

\begin{definition}[Pauli distribution]\label{def:paulidistribution} Let $\ket{\psi}$ be a pure quantum state. The Pauli distribution $p_{\psi}$ is a distribution over the Pauli group $\mathbb{P}_n$, with components  defined as $p_{\psi}(P)=\frac{\tr^2(P\psi)}{d}$.
\end{definition}

\begin{definition}[Bell difference distribution]\label{def:belldifferencedistribution} Let $\ket{\psi}$ be a pure quantum state and $\{\ket{P}\coloneqq\mathbb{I}\otimes P\ket{I}\}$, with $\ket{I}\coloneqq\frac{1}{2^{n/2}}\sum_{i}\ket{i}\otimes\ket{i}$, be the Bell basis on $\mathcal{H}^{\otimes 2}$ labeled by Pauli operators $P\in\mathbb{P}_n$. Let $q_{\psi}(P)\coloneqq|\langle P|\psi^{\otimes 2}\rangle|^2$. The bell difference probability is defined as the following convolution:
\begin{align}
    (q\star q)_{\psi}(P)\coloneqq\sum_{Q}q_{\psi}(Q)q_{\psi}(QP)\,.
\end{align}
Moreover, we define the process of sampling from $(q\star q)_{\psi}(P)$ as Bell difference sampling.
\end{definition}

\begin{lemma}\label{Belldifferenceequalpaulidifference} The Bell difference probability is also given by the convolution between the Pauli distributions
\begin{align}
    (q\star q)_{\psi}(P)=(p\star p)_{\psi}(P)\,,
\end{align}
\begin{proof}
    First notice that we can write $q_{\psi}(P)=\frac{1}{d}\tr(P\psi P\psi^{T})$, where $\psi^T$ denotes the transpose in the computational basis. Notice that given $\psi=\sum_{P}a_{P}P$, then $\psi^T=\sum_{P}a_{P}P^{T}$ where $P^T=\pm P$ depending on the number of $Y$ Pauli matrices. The following direct calculation shows the claim:
    \begin{align}\label{proof1eq1}
        (q\star q)_{\psi}(P)=\sum_{Q}\frac{\tr(Q\psi Q\psi^{T})}{d}\frac{\tr(PQ\psi QP\psi^{T})}{d}=\frac{1}{d^2}\sum_{Q}\tr[(\psi\otimes P\psi P) Q^{\otimes 2}\psi^{T\otimes 2}Q^{\otimes 2}]
    \end{align}
    
We can perform the average over $Q$ and note the following identities 
\begin{align}
\frac{1}{d^2}\sum_{Q}Q^{\otimes 2}\psi^T Q^{\otimes 2}=\frac{1}{d^2}\sum_{Q}\tr^2(Q\psi^{T})Q^{\otimes 2}=\frac{1}{d^2}\sum_{Q}\tr^2(Q\psi)Q^{\otimes 2}=\frac{1}{d^2}\sum_{Q}Q^{\otimes 2}\psi^{\otimes 2}Q^{\otimes 2}
\end{align}
where we used that $\frac{1}{d^2}\sum_{Q}Q^{\otimes 2} P\otimes K Q^{\otimes 2}=\delta_{PK}P\otimes K$ and that $\tr^2(P\psi^T)=\tr^2(P\psi)$. Substituting back to \cref{proof1eq1}, we have
\begin{align}
     (q\star q)_{\psi}(P)=\frac{1}{d^2}\sum_{Q}\tr[(\psi\otimes P\psi P) Q^{\otimes 2}\psi^{\otimes 2} Q^{\otimes 2}]=\frac{1}{d^2}\sum_Q \tr(\psi Q\psi Q)\tr(P\psi PQ\psi Q)=(p\star p)(P)
\end{align}
where in the last line we used that $\tr(\psi Q\psi Q)=\tr^2(Q\psi)$ and that $\tr(P\psi PQ\psi Q)=|\tr(PQ\psi)|^2$. 
\end{proof}
\end{lemma}

\begin{fact}[Convolution of Pauli distribution factorizes]\label{fact:convoludistributionfactorizes} Let $\ket{\psi}=\ket{\phi_{n_1}}\otimes \ket{\phi_{n_2}}$. Then $p_{\psi}=p_{\phi_{n_1}}p_{\psi_{n_2}}$ as well as the convolution $(p_{\psi}\star p_{\psi}) =(p_{\phi_{n_1}}\star p_{\phi_{n_1}})(p_{\phi_{n_2}}\star p_{\phi_{n_2}})$.
    \begin{proof}
        The proof descends from the definition of convolution. 
    \end{proof}
\end{fact}

\begin{definition}[$t$-compressible states~\cite{gu_2025}]  $\ket{\psi}$ is a $t$-compressible states if there exist a Abelian subgroup of the Pauli group $G_{\psi}$, with cardinality $|G_{\psi}|=2^{n-t}$, such that $P\ket{\psi}=\pm \ket{\psi}$ for all $P\in G_{\psi}$.
\end{definition}

\begin{lemma}[Magic compression~\cite{leone_learning_2024,oliviero_unscrambling_2024}] Let $\ket{\psi_t}$ be a $t$-compressible state. There exists a Clifford operation $C\in\mathcal{C}_n$ such that
\begin{align}
    C\ket{\psi_t}=C(\ket{\phi_t}\otimes\ket{0}^{\otimes n-t})
\end{align}
\end{lemma}

\begin{lemma}[Magic compression theorem~\cite{leone_learning_2024,oliviero_unscrambling_2024}]\label{lem:magicompressionth}
    Let $U_t$ be a unitary containing at most $t$ single qubit non-Clifford gates, then $U_t\ket{0}$ is a $t'$-compressible state with $t'\le t$.
\end{lemma}
Next, we will show a very well known fact.
\begin{lemma}\label{lem:haarexpsmall}
    Let $\psi\sim\haar$ a state uniformly sampled from the Haar measure. Then
    \begin{align}
        \Pr_{\psi\sim\haar}\left(\tr^{2}(P\psi)\ge d^{-1/4}\,:\, \forall P\in\mathbb{P}_n\right)\le d^{2}e^{-\Omega(\sqrt{d})}
    \end{align}
    \begin{proof}
The proof uses Levy's lemma~\cite{popescu_entanglement_2006}. First, we have $\mathbb{E}_{\psi\sim\haar}\tr^2(P\psi)=\frac{1}{d+1}$. Then, using that $\tr^2(P\psi)$ has a Lipschitz constant upper bounded by $2$ with respect to $\psi$, we apply Levy's lemma:
\begin{align}
    \Pr_{\psi\sim\haar}\left(\left|\tr^{2}(P\psi)-\frac{1}{d+1}\right|\ge \varepsilon\right)\le e^{-\Omega(\varepsilon^2d)}
\end{align}
choosing $\varepsilon=d^{-1/4}$ and using the union bound over the whole Pauli group $\mathbb{P}_n$ the statement follows.
    \end{proof}
\end{lemma}

\begin{lemma}[Commuting with the support implies large expectation value]\label{lem:commutationexpectationlemma} Let $p_{\psi}$ be the Pauli distribution, and $p_{\psi}\star p_{\psi}$ the convolution with itself. Given $P\in\mathbb{P}_n$, define $\pi_{p_{\psi}}(P)$ and $\pi_{(p\star p)_{\psi}}(P)$ as
\begin{align}
    \pi_{p_{\psi}}(P)\coloneqq\Pr_{Q\sim p_{\psi}}\left([P,Q]=0\right),\quad \pi_{(p\star p)_{\psi}}\coloneqq\Pr_{Q\sim (p\star p)_{\psi}}\left([P,Q]=0\right)
\end{align}
the probability that a Pauli $P$ commute with Pauli operators sampled according either $p_{\psi}$ or $(p\star p)_{\psi}$ respectively. It holds that
\begin{align}
\pi_{p_{\psi}}(P)&=\frac{1+\tr^2(P\psi)}{2}\,,\label{eq1prooflem}\\
\pi_{(p\star p)_{\psi}}(P)&=\frac{1+\tr^4(P\psi)}{2}\,.\label{eq2prooflem}
\end{align}
\begin{proof}
    \cref{eq1prooflem} follows immediately from:
    \begin{align}
         \tr^2(P\psi)=\tr(P\psi P\psi)=\frac{1}{d}\sum_{Q\,:\,[P,Q]=0}\tr^2(Q\psi)-\frac{1}{d}\sum_{Q\,:\,\{P,Q\}=0}\tr^2(Q\psi)=2\pi_{p_{\psi}}(P)-1\,.
    \end{align}
To show \cref{eq2prooflem}, notice that given $P_1,P_2\sim p_{\psi}$, for $Q=P_1P_2$ then $[Q,P]=0$ if $[Q,P_1]=[Q,P_2]=0$ or $\{Q,P_1\}=\{Q,P_2\}=0$, meaning that we have the following relation between $\pi_{p_\psi}$ and $\pi_{(p\star p)_{\psi}}$:
\begin{align}
    \pi_{(p\star p)_{\psi}}=\pi_{p_{\psi}}^2+(1-\pi_{p_{\psi}})^2\label{eq3lemmaproof}
\end{align}
Combining \cref{eq3lemmaproof} and \cref{eq1prooflem} leads to \cref{eq2prooflem}.
\end{proof}
\end{lemma}

\begin{theorem}[Lower bound on state $k$-design via compressible states]\label{th:lowerbound}
    Let $\mathcal{S}_t=\{\ket{\psi_t}\}$ be a set of $t$-compressible states. Let $t\le n-1$. If $k\ge 12t+8$, then it holds that:
    \begin{align}
        \|\mathbb{E}_{\psi\sim\mathcal{S}_t}\psi^{\otimes k}-\mathbb{E}_{\psi\sim\haar}\psi^{\otimes k}\|_1\ge\frac{1}{8}-2^{-\Omega(n)}
    \end{align}
\end{theorem}
\begin{proof}
In what follows, we construct a bounded operator $\Lambda$ which differentiate the two ensembles of states. For this purpose, we describe a computationally efficient distinguishing algorithm for the two ensembles and, at the end, show that the expectation value differences of the resulting POVM ($\frac{\Lambda+\mathbb1}{2}$) is $\Omega(1)$.

   The algorithm proceeds as follows. Let $l \in \mathbb{N}$ be an integer to be determined later. We collect $4l$ copies of the state $\ket{\psi}$, drawn from one of the two ensembles. From these, we perform $l$ independent Bell different measurements (see \cref{def:belldifferencedistribution}), yielding $l$ Pauli operators. For simplicity, we label them in their bitstring representation~\cite{aaronson_improved_2004}: $X = (x_1, \ldots, x_l)$. We define 
    \begin{align}
        S\coloneqq\operatorname{span}(X)\cap X^{\perp}
    \end{align}
where $X^{\perp}\coloneqq\{y\in\mathbb{F}^{2n}_2 : [y,x]=0 \ \forall x\in X \}$, i.e., the set of all Pauli operators commuting with every element of $X$, and $\operatorname{span}$ denotes the span over the field $\mathbb{F}_{2}^{2n}$. The algorithm then selects uniformly at random a nontrivial element from $S$ and outputs outcome coming from the operator, $\tr(P_x^{\otimes 2}\psi^{\otimes2})=\tr^2(P_x \psi)$ for $x\in S$, which corresponds to and operator with bounded operator norm. If $S$ contains only the trivial element (the identity Pauli operator), the algorithm outputs $0$.  

Let us now analyze the performance of the algorithm.

By \cref{lem:haarexpsmall}, if the state is Haar random distributed, the largest Pauli expectation value is exponentially small in $n$ with probability $1-d^{2}e^{-\Omega(\sqrt{d})}$. Let $\Lambda$ be the (efficient) POVM corresponding to the algorithm described above. We have
\begin{align}
    \tr(\Lambda\mathbb{E}_{\psi\sim\haar}\psi^{\otimes 4l+2})\le \left(1-d^{2}e^{-\Omega(d)}\right)d^{-1/4}+d^{2}e^{-\Omega(d)}=2^{-\Omega(n)}
\end{align}
Let us now analyze the case for the ensemble $\mathcal{E}_t$ of $t$-compressible states. We divide the proof into two steps. First we show that $S$ contains (with high probability) at least a non-trivial element once $l> 2t$. Second, we show that every Pauli operator in $S$ has expectation value large.

For the first claim, notice that without loss of generality, we can analyze the case where $\ket{\psi}=\ket{\phi}\otimes \ket{0}^{\otimes n-t}$, as the algorithm we are using is Clifford invariant.

The Pauli distribution (see \cref{def:paulidistribution}) of a state $\ket{\phi}\otimes \ket{0}^{\otimes n-t}$ has support of Pauli of the form $P_t\otimes Z_{n-t}$ where $P_t\in\mathbb{P}_t$ is a Pauli operator on $t$ qubits, while $Z_{n-t}\in\mathbb{Z}_{n-t}$ is a Pauli-Z-string operator on $n-t$ qubits. Moreover, any Pauli of the form $I_t\otimes Z_{n-t}$ belongs to the stabilizer group of $\ket{\psi}$ and therefore, by definition, to $X^{\perp}$. We now show that $\operatorname{span}(X)$ contains one element of the form $I_t\otimes Z_{n-t}$, which implies that $S$ contains more than the trivial element. Combining \cref{Belldifferenceequalpaulidifference} and \cref{fact:convoludistributionfactorizes}, the samples $P_{t}\otimes Z_{n-t}$ belonging to $X$ are sampled {\em independently}, i.e. $P_{t}$ according to $p_{\phi_t}\star p_{\phi_t}$ and $Z_{n-t}$ according to $p_{0_{n-t}}\star p_{0_{n-t}}$. Moreover, since the convolution of two uniform distribution is still a uniform distribution, $p_{0_{n-t}}\star p_{0_{n-t}}$ is uniformly distributed on $\mathbb{Z}_{n-t}$. Since $P_t$ and $Z_{n-t}$ are sampled independently, we can analyze them separately.
Denote $X|_t$ the restriction of the samples on the first $t$ qubits. On the one hand, note that if we collect $|X|\ge 2t+1$ samples, then—since the maximal size of an algebraically independent set of Pauli operators on $t$ qubits is $2t$—the restriction of $X$ to the first $t$ qubits, denoted $X|_t$, must contain a linear dependency. Consequently, $I_t$, i.e., the identity Pauli operator (all zeros in the bitstring representation) lies non-trivially in the span of $X|_t$, meaning there exists an $0\neq \alpha\in\mathbb F_2^{l}$, such that 
\begin{align}
    0=\sum_{i=1}^l \alpha_i X_i|_t\,.
\end{align}

On the other hand, since the  $Z_{n-t}$ Pauli operator are uniformly distributed, the probability that also
\begin{align}
    0=\sum_{i=1}^l \alpha_i X_i=0\oplus Z'\,.
\end{align}
holds is negligibly small. In particular, the probability that $Z'$ is all zero is given by $p=2^{-(n-t)}\,$. Therefore, with probability at least $1 - 2^{-(n-t)}$, the set $S$ contains a nontrivial element.

Let us now follow the second step of the proof: we show that, with high probability, every element in $S$ has a large expectation value. Since $P\in S$ commutes with $X$ by definition, we apply \cref{lem:commutationexpectationlemma}. The probability that $l$ samples from $p_{\psi}\star p_{\psi}$ commute with a given $P$ is $(\pi_{p_{\psi}\star p_{\psi}})^{l}$ (see \cref{lem:commutationexpectationlemma} for the definition of $\pi_{p_{\psi}\star p_{\psi}}$). From \cref{eq2prooflem}, if we fix a threshold $\varepsilon_T>0$, the probability that $P\in S$ (which commutes with all samples) has expectation value $\tr^2(P\psi)<\varepsilon_T$ is at most $\left(\frac{1+\varepsilon_T^2}{2}\right)^l$. By the union bound, the probability that all Pauli operators in $S$ have expectation value below $\varepsilon_T$ is at most $4^t\left(\frac{1+\varepsilon_T^2}{2}\right)^l$. The factor $4^t$ arises because $S$ contains Paulis of the form $P=P_t\otimes Z_{n-t}$ and, as such, $\tr^2(P\psi^{\otimes 2})=\tr^2(P_t\phi_t)$, i.e., the squared expectation value depends only on $P_t$ (since $Z_{n-t}$ belongs to the stabilizer group). We can finally lower bound the expectation value of the bounded operator $\Lambda$ over $\mathbb{E}_{\psi\sim\mathcal{E}_t}\psi^{\otimes 4l+2}$:
\begin{align}
    \tr(\Lambda \mathbb{E}_{\psi\sim\mathcal{E}_t}\psi^{\otimes 4l+2})\ge \varepsilon_T\left(1-2^{-(n-t)}-4^{t}\left(\frac{1+\varepsilon_T^2}{2}\right)^l\right)
\end{align}
We can therefore lower bound the trace distance as
\begin{align}
\|\mathbb{E}_{\psi\sim\mathcal{S}_t}\psi^{\otimes t'}-\mathbb{E}_{\psi\sim\haar}\psi^{\otimes t'}\|_1\ge\varepsilon_T\left(1-2^{-(n-t)}-4^{t}\left(\frac{1+\varepsilon_T^2}{2}\right)^l\right)-2^{-\Omega(n)}
\end{align}
For simplicity, let us consider $\varepsilon_T=\frac{1}{2}$. Let us impose the most stringent condition on $t$, i.e. $\le n-1$. Then, we further impose 
\begin{align}
   2^{2t-l\log_2\frac{8}{5}}\le \frac{1}{4}
\end{align}
for which is sufficient to choose $l\ge 3t+2$ and impose $n\ge t+2$ to ensure $\|\mathbb{E}_{\psi\sim\mathcal{S}_t}\psi^{\otimes k}-\mathbb{E}_{\psi\sim\haar}\psi^{\otimes k}\|_1\ge \frac{1}{8}-2^{-\Omega(n)}$. Hence a total number of copies of $k=4l+2=12t+10$ are sufficient to ensure a large trace distance. This concludes the proof.

\end{proof}

As a simple corollary, we have a lower bound that ensures the optimality in terms of non-Clifford support of our diluted $k$-design in \cref{maintheorem}.
\begin{corollary}[Quantum homeopathy dilution is optimal]\label{cor:quantumhomoptimal}
    Let $\mathcal{E}_t \coloneqq \{C_1 U_t C_2\}$ be the ensemble of unitaries with $C_1, C_2 \sim \mathcal{C}_n$ and $U_t$ an exact $k$-design on $t$ qubits. Then $\mathcal{E}_t$ cannot form an additive-error unitary $k$-design, and therefore a quantum-secure design, unless $t \ge (k-8)/12$. Therefore the dilution of a $k$-design on $\Theta(k)$ qubits is necessary and sufficient for the construction of a quantum-secure $k$-design on $n$ qubits.
\begin{proof}
    The proof follows from \cref{th:lowerbound} by noticing that the ensemble $\mathcal{S}_t=\{U\ket{0}\,:\, U\in\mathcal{E}_t\}$ is an ensemble of $t$-compressible states.
\end{proof}
\end{corollary}

As a direct corollary, we have our lower bound for unitary $k$-designs.
\begin{corollary}[$\Omega(k)$ non-Clifford gates are necessary for $k$-designs]\label{cor:lowerboundnoncliffordgates}
  Let $\mathcal{E}_t$ be an ensemble of unitaries containg at most $t<n$ many single qubits non-Clifford gates. Then $\mathcal{E}_t$ cannot form an additive unitary $k$-design $\varepsilon$-approximate $k$-design, and therefore a quantum-secure design, unless $t\ge (k-8)/12$.  
  \begin{proof}
      The claim follows from \cref{lem:magicompressionth,th:lowerbound}.
  \end{proof}
\end{corollary}

\subsection{Implications of our results}\label{sec:implicationsofourresults}
In this section, we discuss the implications of our results. As explained in the main text, our homeopathy construction naturally requires only a system-size–independent number of non-Clifford gates. This is because the only step that involves non-Clifford gates is the injection of a $k$-design on $t = O(k)$ qubits. Here, we examine the best known constructions from previous work to estimate how many non-Clifford gates are sufficient for our approach. 

\begin{lemma}[Shallow unitary designs~\cite{schuster2025randomunitariesextremelylow}]\label{lem:shallowunitarydesigns} There exists a construction for relative $\varepsilon$-approximate unitary $k$-designs on $m$ qubits with unitaries with depth $O(k\log \frac{m}{\varepsilon})$.
\end{lemma}

\begin{lemma}[Shallow unitary design with extra qubits~\cite{cui2025unitarydesignsnearlyoptimal}]\label{lem:shallowandextraqubits} There exists a construction for relative $\varepsilon$-approximate unitary $k$-design on $m$ qubits with depth $O(\log k\log\log m/\varepsilon)$ using $O(km\log\log m/\varepsilon)$ many gates.
    
\end{lemma}

\begin{lemma}[Shallow pseudorandom unitaries~\cite{schuster2025randomunitariesextremelylow}]\label{lem:shallowpseudorandomunitaries} There exists a construction of pseudorandom unitaries on $m$ qubits, secure against any quantum algorithm running in time $\exp(m^{1-c})$ for any $c>0$, with depth $O(\poly\log m)$, under the cryptographic assumption that the Learning with Errors (LWE) problem~\cite{ma2025constructrandomunitaries,schuster2025randomunitariesextremelylow} cannot be solved by any sub-exponential time quantum algorithm.
\end{lemma}

\begin{lemma}[Shallow-depth Clifford unitaries with extra qubits~\cite{moore1998parallelquantumcomputationquantum,jiang2022optimalspacedepthtradeoffcnot}]\label{lem:shallowClifford} Any Clifford operator can be realized in depth $O(\log n)$ and $O(n^2)$ extra qubits with all-to-all connectivity.
\end{lemma}

Thanks to the lemmas above, and our quantum homeopathy construction, we can infere about the number of single qubits non-Clifford gates necessary for the construction of quantum secure unitary designs.

\begin{corollary}[Quantum secure unitary designs with few non-Clifford gates. \cref{cor1} in the main text]\label{cor1app} There exists a construction of quantum secure unitary $k$-designs using $O(k^2\log\varepsilon^{-1}\log (k+\log\varepsilon^{-1})/\varepsilon)$ number of non-Clifford gates. Alternatively, using $\tilde{O}(k^2)$ extra qubits the non-Clifford cost reduces to $O(k^2\log\varepsilon^{-1})$. 
\begin{proof}
In our quantum homeopathy construction, we can employ a exact $k$-design on $2k+6+\log\varepsilon^{-1}$ qubits. We can however substitute it with a $\varepsilon$-approximate $k$-designs. According to \cref{th:maintheorem1}, we need a $k$-design on $2k+\log2\varepsilon^{-1}+6$ qubits to have a quantum secure $\varepsilon/2$ unitary $k$-design diluted on $n$ qubits. Using \cref{lem:relativeimpliesquantumsecurity}, we can substitute the exact $k$-design on $2k+\log2\varepsilon^{-1}+6$ with a $\varepsilon/2$ quantum secure design requiring depth (\cref{lem:shallowunitarydesigns}) $O(k\log4[(2k+\log2\varepsilon^{-1}+6)/\varepsilon])$ and, therefore, $O((2k^2+k\log2\varepsilon^{-1}+6k)\log4[(2k+\log2\varepsilon^{-1}+6)/\varepsilon])=O(k^2\log k/\varepsilon)$ gates. Finally, the use of triangle inequality shows the first claim. Similarly, we can employ extra qubits, as in \cref{lem:shallowandextraqubits}. Using a similar technique, it yields to a non-Clifford gates count of $\tilde{O}(k^2\log\varepsilon^{-1})$ up to log factors and $\tilde{O}(k^2)$ extra qubits. 
\end{proof}
\end{corollary}

According to our lower bound, the optimal number of non-Clifford gates required to construct a quantum-secure design is lower bounded by $\Omega(k)$. Hence, our construction appears to be loose by a quadratic factor. However, a closer examination reveals that: (i) the lower bound explicitly proves the existence of an \textit{computationally efficient} quantum algorithm to distinguish the two ensembles, so a fair comparison requires considering this aspect; and (ii) we can employ pseudorandom unitaries to guarantee quantum-polynomial security as defined in \cref{def2app}. Thanks to \cref{lem:shallowpseudorandomunitaries}, the following corollary shows that our construction is optimal for quantum-polynomially-secure unitary $k$-designs, achieving the construction using only $\tilde{O}(k)$ non-Clifford gates.

\begin{corollary}[Quantum-polynomially-secure-unitary $k$-designs with optimal number of non-Clifford gates. \cref{cor2} in the main text]\label{cor2app} Let $c>0$. Under the computational assumption that LWE cannot be solved in sub-exponential time by a quantum computer, for any $\varepsilon=\exp(o(n))$, there exists a construction of quantum-polynomially-secure unitary $k$-design using $\tilde{O}((k+\log^{1+c}\varepsilon^{-1}+\log^{1+c}n)\poly\log (k+\log^{1+c}\varepsilon^{-1}+\log^{1+c}n))$ number of non-Clifford gates. For $k=\omega(\log n)$, this number is also optimal up to log factors. 
\begin{proof}
    First, we notice that a pseudorandom unitary on $m$ qubits is, under computational assumption in \cref{lem:shallowpseudorandomunitaries}, a quantum-polynomially-secure $k$-design for all $k=O(\exp(m^{1-c}))$ and $\varepsilon(m)=\Omega(\exp(-m^{1-c}))$. We aim to replace the exact $k$-design on $t\ge 2k+\log\varepsilon^{-1}+6$ qubits in our homeopathy construction with a pseudorandom unitary. However, under computational assumption in \cref{lem:shallowpseudorandomunitaries}, pseudorandom unitaries are only guaranteed to be secure against $\poly(n)$ time adversaries if the subsystem $m=\omega(\log n)$. We have to distinguish several cases, but the proof proceeds analogously to the proof of \cref{cor1app}, i.e. we just replace a pseudorandom unitary on a portion of qubits and use triangle inequality to show that the construction remains indistinguishable from Haar random for any $k$-query quantum algorithm with resolution at most $\varepsilon$.
    \begin{itemize}
        \item $k+\log\varepsilon^{-1}=o(\log^{1+c} n)$: in order to guarantee security against any $\poly(n)$ time quantum algorithm, we place a pseudorandom unitary on $m=O(\log^{1+c}n)$ and then dilute it with our homeopathy construction. The number of gates is $O(\log^{1+c}n\poly\log\log n )$.
        \item $k+\log\varepsilon^{-1}=\Omega(\log^{1+c}n)$. In this case, depending on the target accuracy $\varepsilon$, we distinguish two scenarios:
         \begin{itemize}
        \item { $\varepsilon=\Omega(\exp(-[2k+\log\varepsilon^{-1}]^{1-c})$}: we place a pseudorandom unitary on $2k+\log\varepsilon^{-1}+6$ qubits which, thanks to \cref{lem:shallowpseudorandomunitaries}, can be constructed using $O((k+\log\varepsilon^{-1})\poly\log(k+\log\varepsilon^{-1}))$ many gates, which shows the claim.
        \item $\varepsilon=o(\exp(-[2k+\log\varepsilon^{-1}]^{1-c})$: we define $t_{\text{hom}}\coloneqq2k+\log\varepsilon^{-1}+6$, and inject a pseudorandom unitary on $t=t_{\text{hom}}+t_{\varepsilon}$ qubits, where the extra qubits $t_{\varepsilon}$ are needed to reach the desired target resolution $\varepsilon$. Indeed, since a pseudorandom unitary on $t$ qubits is secure against algorithms with resolution $O(-\exp(t^{1-c}))$, we impose
        \begin{align}
            \varepsilon=\Omega(\exp[-(t_{\text{hom}}+t_{\varepsilon})^{1-c}])
        \end{align}
    which imply that additional extra qubits $t_{\varepsilon}=O(\log^{1+c}\varepsilon^{-1})$ are sufficient. As before, such a pseudorandom unitary can be constructed with $O((t_{\text{hom}}+\log^{1+c}\varepsilon^{-1})\poly\log(t_{\text{hom}}+\log^{1+c}\varepsilon^{-1}))=\tilde{O}(k+\log^{1+c}\varepsilon^{-1})$ many gates.
    \end{itemize}
        \end{itemize}
   
Putting together all the cases, we see that for any target resolution $\varepsilon$ the injection of a pseudorandom unitary on $O(k+\log^{1+c}\varepsilon^{-1}+\log^{1+c} n)$ qubits is sufficient for constructing a $\varepsilon$-approximate quantum-polynomially-secure unitary $k$-design and requires a number of gates is given by $O((k+\log^{1+c}\varepsilon^{-1}+\log^{1+c}n)\poly\log(k+\log^{1+c}\varepsilon^{-1}+\log^{1+c}n))$ showing the claim. The optimality (with respect to the linear scaling with $k$) follows directly from the lower bound in \cref{sec:lowerbound}.
\end{proof}
    
\end{corollary} 


Let us now show that our construction can be implemented in shallow depth using extra qubits. 
\begin{corollary}[Shallow-depth and few non-Clifford gates quantum secure design. \cref{cor3} in the main text] For any $\varepsilon$, there exists a construction of quantum-secure unitary $k$-design using $\tilde{O}(k^2\log\varepsilon^{-1})$ non-Clifford gates, $O(k^2+n^2)$ extra qubits in depth $\tilde{O}(\log k+\log n)$.
\begin{proof}
    The proof makes use of \cref{lem:shallowandextraqubits,lem:shallowClifford} and proceeds in a similar way of \cref{cor1app,cor2app}.
\end{proof}
    
\end{corollary}
\bibliographystyle{apsrev4-1}
\bibliography{bib3} 

\end{document}